\documentclass[11pt]{article}
\usepackage{jeffstyle}
\usepackage{amsopn,amssymb,amsthm,amsmath}
\usepackage{color}
\usepackage{xspace}
\usepackage[noend]{algorithmic}
\usepackage{algorithm,graphicx,tikz}
\usepackage{euscript}

\newcommand{\R}{\mathbb{R}}
\renewcommand{\H}{\EuScript{H}}
\newcommand{\Nyst}{\textsc{Nystr\"{o}m}\xspace}
\newcommand{\RFM}{\textsc{RFFMaps}\xspace}

\newlength{\figsize} 

\title{Streaming Kernel Principal Component Analysis}

\author{
Mina Ghashami\footnotemark[1]\footnote{These authors contributed equally to the paper.}\\
School of Computing\\
University of Utah\\
\texttt{ghashami@cs.utah.edu} \\
\and
Daniel Perry\footnotemark[1] \\
School of Computing\\
University of Utah\\
\texttt{dperry@cs.utah.edu} \\
\and
Jeff M. Phillips \\
School of Computing\\
University of Utah\\
\texttt{jeffp@cs.utah.edu} \\
}

\date\nonumber
\begin{document}

\maketitle

\begin{abstract}
Kernel principal component analysis (KPCA) provides a concise set of basis vectors which capture non-linear structures within large data sets, and is a central tool in data analysis and learning.  To allow for non-linear relations, typically a full $n \times n$ kernel matrix is constructed over $n$ data points, but this requires too much space and time for large values of $n$.  Techniques such as the Nystr\"{o}m method and random feature maps can help towards this goal, but they do not explicitly maintain the basis vectors in a stream and take more space than desired.
We propose a new approach for streaming KPCA which maintains a small set of basis elements in a stream, requiring space only logarithmic in $n$, and also improves the dependence on the error parameter.  Our technique combines together random feature maps with recent advances in matrix sketching, it has guaranteed spectral norm error bounds with respect to the original kernel matrix, and it compares favorably in practice to state-of-the-art approaches. 
\end{abstract}

\section{Introduction}
\label{sec:intro}

Principal component analysis (PCA) is a well-known technique for dimensionality reduction, and has many applications including visualization, pattern recognition, and data compression \cite{jolliffe2005principal}.   
Given a set of centered $d$-dimensional (training) data points $A = [a_1;\ldots; a_n] \in \R^{n \times d}$, PCA diagonalizes the covariance matrix $C = \frac{1}{n}A^TA$ by solving the eigenvalue equation $C v = \lambda v$.
However, when the data points lie on a highly nonlinear space, PCA fails to concisely capture the structure of data.
To overcome this, several nonlinear extension of PCA have been proposed, in particular Kernel Principal Component Analysis (KPCA)~\cite{scholkopf1997kernel}.
The basic idea of KPCA is to implicitly map the data into a nonlinear feature space of high (or often infinite) dimension and perform PCA in that space~\cite{scholkopf1997kernel}. The nonlinear map is often denoted as $\phi:\R^d \rightarrow \H$ where $\H$ is a Reproducing Kernel Hilbert Space (RKHS). While direct computation of PCA in RKHS is infeasible, we can invoke the so called \textit{kernel trick} which exploits the fact that PCA interacts with data through only pair-wise inner products.  That is $\langle \phi(x), \phi(y)\rangle_{\H} = K(x,y)$, for all $x,y\in \R^d$ for a kernel function $K$; we represent this as the $n \times n$ gram matrix $G$.  
However, KPCA suffers from high space and computational complexity in storing the entire kernel (gram) matrix $G \in \R^{n \times n}$ and in computing the decomposition of it in the training phase.  Then in the testing phase it spends $O(nd)$ time to evaluate the kernel function for any arbitrary test vector with respect to all training examples.
Although one can use low rank decomposition approaches~\cite{drineas2006fast2,s06,l13,gp14} to reduce the computational cost to some extent, KPCA still needs to compute and store the kernel matrix.

There have been two main approaches towards resolving this space issue.  
First approach is the Nystr\"{o}m \cite{williams2001using} which uses a sample of the data points to construct a much smaller gram matrix.  
Second approach is using feature maps \cite{rahimi2007random} which provide an approximate but explicit embedding of the RKHS into Euclidean space.  
As we describe later, both approaches can be made to operate in a stream, approximating the KPCA result in less than $O(n^2)$ time and space.  
%
Once these approximations are formed, they reveal a $D \ll n$ dimensional space, and typically a $k$-dimensional subspace found through linear PCA in $\R^D$, which captures most of the data (e.g., a low rank-$k$ approximation).  
There are two main purposes of these $D$- and $k$-dimensional subspaces; they start with mapping a data point $x \in \R^d$ into the $D$-dimensional space, and then often onto the $k$-dimensional subspace. If $x$ is one of the training data points, then the $k$-dimensional representation can be used as a concise ``loadings'' vector.  It can be used in various down-stream training and learning tasks wherein this $k$-dimensional space, can assume linear relations (e.g., linear separators, clustering under Euclidean distance) since the non-linearity will have already been represented through the mapping to this space.  
If $x$ is not in the training set, and the training set represents some underlying distribution, then we can assess the ``fit'' of $x$ to this distribution by considering the residual of its representation in the $D$-dimensional space when projected to the $k$-dimensional space.  

We let \textsc{Test time} refer to this time for mapping a single point $x$ to the $D$-dimensional and $k$-dimensional spaces.  The value of $k$ needed to get a good fit depends on the choice of kernel and its fit to the data; but $D$ depends on the technique.  For instance in (regular) KPCA $D=n$, in Nystr\"{o}m $D = O(1/\eps^2)$, when using random feature maps with \cite{rahimi2007random} $D = O((1/\eps^2)\log n)$, where $\eps \in (0,1)$ is the error parameter.  
We propose a new streaming approach, named as SKPCA, that will only require $D = O(1/\eps)$.

We prove bounds and show empirically that SKPCA greatly outperforms existing techniques in \textsc{Test time}, and is also comparable or better in other measures of \textsc{Space} (the cost of storing this map, and space needed to construct it), and \textsc{Train time} (the time needed to construct the map to the $D$-dimensional and $k$-dimensional spaces).  

\paragraph{Background and Notation.}
We indicate matrix $A$ is $n \times d$ dimensional as $A \in \R^{n \times d}$.   
Matrices $A$ and $Z$ will be indexed by their row vectors $A = [a_1; a_2; \ldots, a_n]$ while other matrices $V, U, W, \ldots$ will be indexed by column vectors $V = [v_1, v_2, \ldots, v_d]$.  
We use $I_n$ for the $n$-dimensional identity matrix and $0^{n \times d}$ as the full zero matrix of dimension $n \times d$.
The Frobenius norm of a matrix $A$ is $\|A\|_F = \sqrt{\sum_{i=1} \|a_i\|^2}$ and the spectral norm is $\|A\|_2 = \sup_{x \in \R^d} \frac{\|Ax\|}{\|x\|}$. 
We denote transpose of a matrix as $A^T$.
The singular value decomposition of matrix $A \in \mathbb{R}^{n \times d}$ is denoted by $[U,S,V] = \svd(A)$. 
If $n \ge d$ it guarantees that $A = U S V^T$, $U^TU = I_n$, $V^TV = I_d$, $U\in \R^{n \times n}$, $V\in \R^{d \times d}$,
and $S = \diag(s_1, s_2, \ldots, s_d) \in \R^{n \times d}$ is a diagonal matrix with $s_1 \ge s_2 \ge \ldots \ge s_d \ge 0$. 
Let $U_k$ and $V_k$ be matrices containing the first $k$ columns of $U$ and $V$, respectively, and $S_k =\diag(s_1, s_2, \ldots, s_k) \in \R^{k \times k}$.
The matrix $A_k = U_k S_k V_k^T$ is the best rank $k$ approximation of $A$ in the sense that $A_k = {\arg \min}_{C : \rank(C) \leq k} \|A - C\|_{2,F}$.  
We denote by $\pi_B(A)$ the projection of rows of $A$ on the span of the rows of $B$. In other words, $\pi_B(A) = A B^\dagger B$ where $(\cdot)^\dagger$ indicates taking the Moore-Penrose psuedoinverse.
Finally, expected value of a matrix is defined as the matrix of expected values.

\subsection{Related Work}

\paragraph{Matrix Sketching.}
Among many recent advancements in matrix sketching~\cite{Woo14,Mah11}, we focus on those that compress a $n \times d$ matrix $A$ into an $\ell \times d$ matrix $B$.  
There are several classes of algorithms based on row/column sampling \cite{drineas2006fast2,BMD09} (very related to Nystr\"{o}m approaches \cite{drineas2005nystrom}), random projection \cite{s06} or hashing \cite{clarkson2013low} which require $\ell \approx \textsf{c}/\eps^2$ to achieve $\eps$ error.  The constant $\textsf{c}$ depends on the algorithm, specific type of approximation, and whether it is a ``for each'' or ``for all'' approximation.  
A recent and different approach, Frequent Directions (FD)~\cite{l13}, uses only $\ell=2/\eps$ to achieve the error bound $\|A^T A - B^T B\|_2 \leq \eps \|A\|_F^2$, and runs in time $O(nd/\eps)$.  
We use a modified version of this algorithm in our proposed approach.  

\paragraph{Incremental Kernel PCA.}
Techniques of this group update/augment the eigenspace of kernel PCA without storing all training data.
\cite{kimura2005incremental} adapted incremental PCA~\cite{hall1998incremental} to maintain a set of linearly independent training data points and compute top $d$ eigenvectors such that they preserve a $\theta$-fraction (for a threshold $\theta \in (0,1)$) of the total energy of the eigenspace. 
However this method suffers from two major drawbacks.  First, the set of linearly independent data points can grow large and unpredictably, perhaps exceeding the capacity of the memory. Second, under adversarial (or structured sparse) data, intermediate approximations of the eigenspace can compound in error, giving bad performance~\cite{GDP14}.  
%
Some of these issues can be addressed using online regret analysis assuming incoming data is drawn iid (e.g., ~\cite{kuzmin2007online}).
However, in the adversarial settings we consider, FD~\cite{l13} can be seen as the right way to formally address these issues.  

\paragraph{Nystr\"{o}m-Based Methods for Kernel PCA.}
Another group of methods~\cite{williams2001using,drineas2005nystrom,gittens2013revisiting,kumar2012sampling,talwalkar2010matrix}, known as \Nyst, approximate the kernel (gram) matrix $G$ with a low-rank matrix $\hat G$, by sampling columns of $G$.
The original version~\cite{williams2001using} samples $c$ columns with replacement as $C$ and estimates $\hat G = C W^{-1} C^T$, where $W$ is the intersection of the sampled columns and rows; this method takes $O(nc^2)$ time and is not streaming.  
Later \cite{drineas2005nystrom} used sampling with replacement and approximated $G$ as $\hat G = C W_k^{\dagger} C^T$. They proved if sampling probabilities are of form $p_i = G_{ii}^2/\sum_{i=1}^n G_{ii}^2$, then for $\eps\in (0,1)$ and $\delta\in(0,1)$, a Frobenius error bound
$\|G - \bar G_k\|_F \leq \|G - G_k\|_F + \eps \sum_{i=1}^n G_{ii}^2$
holds with probability $1-\delta$ for $c = O((k/\eps^4)\log(1/\delta))$, and a spectral error bound
$\|G - \bar G_k\|_2 \leq \|G - G_k\|_2 + \eps \sum_{i=1}^n G_{ii}^2$
holds with probability $1-\delta$ for $c = O((1/\eps^2)\log(1/\delta))$ samples.
There exist conditional improvements, e.g., \cite{gittens2013revisiting} shows with $c = O(\mu \frac{k \ln (k/\delta)}{\eps^2})$ where $\mu$ denotes the coherence of the top $k$-dimensional eigenspace of $G$, that $\|G - \bar G_k\|_2 \leq (1+\frac{n}{(1-\eps)c})\|G - G_k\|_2 $.

\begin{table*}[t!]
\caption{Asymptotic \textsc{Train time}, \textsc{Test time} and \textsc{Space} for \SKPCA, \textsc{KPCA}~\cite{scholkopf1997kernel} , \textsc{RNCA}~\cite{lopez2014randomized}, and \Nyst~\cite{drineas2005nystrom} to achieve $\|G' - G\|_2 \leq \eps n$ with high probability (KCPA is exact) and with Gaussian kernels.}

\vspace{-.1in}
 
\label{tbl:compare}
\begin{center}
\begin{small}
\begin{sc}
\begin{tabular}{r|c|c|c}
\hline
& Train time & Test time & Space \\
\hline 
KPCA 
 & $O(n^2d + n^3)$ & $O(n^2 + nd)$ & $O(n^2 + nd)$ \\
Nystr\"{o}m 
 & $O((n \log n)/\eps^2 + (d \log^3 n)/\eps^4 + (\log^3 n)/\eps^6) $ 
 & $O((d \log n)/\eps^2 + (\log^2 n)/\eps^4)$ 
 & $O((d \log n)/\eps^2 + (\log^2 n)/\eps^4)$\\ 
RNCA 
 & $O((n d\log n)/\eps^2 + (n\log^2 n)/\eps^4 + (\log^3 n)/\eps^6)$ 
 & $O((d\log n)/\eps^2 + (k\log n)/\eps^2)$ 
 & $O((d\log n)/\eps^2 + (\log^2 n)/\eps^4)$ \\
\hline
SKPCA 
& $O((n d \log n)/\eps^2 + (n \log n)/\eps^3)$ 
& $O((d \log n)/\eps^2 + (k\log n)/\eps^2)$ 
& $O((d\log n)/\eps^2 + (\log n)/\eps^3)$ \\
\hline
\end{tabular} 
\end{sc}
\end{small}
\end{center}
\vspace{-.2in}
\end{table*}

\paragraph{Random Fourier Features for Kernel PCA.}
In this line of work, the kernel matrix is approximated via randomized feature maps. 
The seminal work of \cite{rahimi2007random} showed one can construct randomized feature maps $Z: \R^d \to \R^m$ such that for any shift-invariant kernel $K(x,y) = K(x-y)$ and all $x,y\in\R^d$,  $\E[\langle Z(x),Z(y) \rangle] = K(x,y)$ and if $m = O((d/\eps^2) \log (n/\delta))$, then with probability at least $1-\delta$,  
$\left|\langle Z(x),Z(y) \rangle - K(x,y)\right| \leq \eps$.
Using this mapping, instead of \emph{implicitly} lifting data points to $\H$ by the kernel trick, they \emph{explicitly} embed the data to a low-dimensional Euclidean inner product space.
Subsequent works generalized to other kernel functions such as group invariant kernels~\cite{li2010random}, min/intersection kernels~\cite{maji2009max}, dot-product kernels~\cite{kar2012random}, and polynomial kernels~\cite{hamid2013compact,ANW14}. 
This essentially converts kernel PCA to linear PCA. 
In particular, Lopez \etal~\cite{lopez2014randomized} proposed \textsc{randomized nonlinear PCA (RNCA)}, which is an exact linear PCA on the approximate data feature maps matrix $Z\in \R^{n \times m}$.  
They showed the approximation error is bounded as $\E [\|\hat G - G\|_2] \leq \Theta((n \log{n})/m)$, where $\hat G = Z Z^T$ is not actually constructed.

\subsection{Our Result vs. Previous Streaming}

In this paper, we present a streaming algorithm for computing kernel PCA where the kernel is any shift-invariant function $K(x,y) = K(x-y)$. We refer to our algorithm as \SKPCA (Streaming Kernel PCA) throughout the paper.  Transforming the data to a $m$-dimensional random Fourier feature space $Z\in \R^{n\times m}$ (for $m \ll n$) and maintaining an approximate $\ell$ dimensional subspace $W\in \R^{m\times \ell}$ ($\ell \ll m$), we are able to show that for $\tilde G = Z W W^T Z^T$, the bound $\|G - \tilde G\|_2 \leq \eps n$ holds with high probability.  
Our algorithm requires $O(dm + m\ell)$ space for storing $m$ feature functions and the $\ell$-dimensional eigenspace $W$. 
Moreover \SKPCA needs $O(dm + ndm + nm\ell) = O(nm(d + \ell))$ time to compute $W$, and permits $O(dm + m\ell)$ test time to first transfer the data point to $m$-dimensional feature space and then update the eigenspace $W$.  

We compare with two streaming algorithms; \textsc{RNCA}\cite{lopez2014randomized} and  \Nyst~\cite{drineas2005nystrom}.

\textsc{RNCA}\cite{lopez2014randomized} achieves $\E [\|\hat G - G\|_2] \leq \Theta((n \log{n})/m)$, for $\hat G = Z Z^T$ and $Z\in \R^{n \times m}$ being the data feature map matrix. Using Markov's inequality it is easy to show that with any constant probability $\|\hat G - G\|_2 \leq \eps n$ if $m  = O((\log n)/\eps)$.   We extend this (in Lemma \ref{lem:HiltoRR}) to show with $m = O((\log n)/\eps^2)$ then $\|\hat G - G\|_2 \leq \eps n$ with high probability $1-1/n$. 
This algorithm takes $O(dm + nmd + nm^2) = O(nmd + nm^2)$ time to construct $m$ feature functions, apply them to $n$ data points and compute the $m\times m$ covariance matrix (adding $n$ $m \times m$ outer products). Moreover, it takes $O(dm + m^2)$ space to store the feature functions and covariance matrix. Testing on a new data point $x \in \R^d$ is done by applying the $m$ feature functions on $x$, and projecting to the rank-$k$ eigenspace in $O(dm+ mk)$.

\Nyst~\cite{drineas2005nystrom} approximates the original Gram matrix with $\bar G = C W_k^{\dagger} C^T$.
For shift-invariant kernels, the sampling probabilities are $p_i = G_{ii}^2/\sum_{i=1}^n G_{ii}^2 = 1/n$, hence one can construct $W\in\R^{c\times c}$ in a stream using $c$ independent reservoir samplers.  
Note setting $k = n$ (hence $G_k = G$), their spectral error bound translates to
$\|G - \bar G\|_2 \leq \eps n$ for $c = O((1/\eps^2)\log(1/\delta))$.  
Their algorithm requires $O(nc + dc^2\log n)$ time to do the sampling and construct  $W$. It also needs $O(cd + c^2)$ space for storing the samples and $W$.  
The test time step on a point $x \in \R^d$ evaluates $K(x,y)$ on each data point $y$ sampled, taking $O(cd)$ time, 
and projects onto the $c$-dimensional and $k$-dimensional basis in $O(c^2 + ck)$ time; requiring $O(cd + c^2)$ time.  

For both RNCA and \Nyst we calculate the eigendecomposition once at cost $O(m^3)$ or $O(c^3)$, respectively, at \textsc{Train time}.  Since \emph{\SKPCA maintains this decomposition at all steps}, and \textsc{Test}ing may occur at any step, this favors RNCA and \Nyst.   

Table \ref{tbl:compare} summarizes train/test time and space usage of above mentioned algorithms. KPCA is included in the table as a benchmark. All bounds are mentioned for high probability $\delta = 1/n$ guarantee. As a result $c = O((\log n)/\eps^2)$ for \textsc{Nystr\"{o}m} and $m = O((\log n)/\eps^2)$ for \textsc{RNCA} and \SKPCA. 
One can use Hadamard fast Fourier transforms (Fastfood)\cite{le2013fastfood} in place of Gaussian matrices to gain improvement on train/test time and space usage of \SKPCA and \textsc{RNCA}. These matrices allow us to compute random feature maps in time $O(m \log d)$ instead of $O(md)$, and to store feature functions in space $O(m)$ instead of $O(md)$.  
Since $d$ was relatively small in our examples, we did not observe much empirical benefit of this approach, and we omit it from our further discussions.  

We see \SKPCA wins on \textsc{Space} and \textsc{Train time} by factor $(\log n)/\eps$ over RNCA, and similarly on
\textsc{Space} and \textsc{Test time} by factors $(\log n)/\eps$ and $(\log n)/(\eps^2 k)$ over \Nyst.  
When $d$ is constant and $\eps < ((\log^2 n)/n)^{1/3}$ it improves \textsc{Train time} over \Nyst.  
It is the first method to use \textsc{Space}  sublinear (logarithmic) in $n$ and sub-quartic in $1/\eps$, and have \textsc{Train time} sub-quartic in $1/\eps$, even without counting the eigen-decomposition cost.

\section{Algorithm and Analysis}
\label{sec:RFF}
In this section, we describe our algorithm \textsc{Streaming Kernel Principal Component Analysis (\SKPCA)} for approximating the eigenspace of a data set which exists on a nonlinear manifold and is received in streaming fashion one data point at a time.
\SKPCA consists of two implicit phases. In the first phase, a set of $m$ data oblivious random feature functions ($f_1,\cdots,f_m$) are computed to map data points to a low dimensional Euclidean inner product space. 
These feature functions are used to map each data point $a_i \in \R^d$ to $z_i \in \R^m$.  
In the second phase, each approximate feature vector $z_i$ is fed into a modified \FD~\cite{l13} which is a small space streaming algorithm  for computing an approximate set of singular vectors; the matrix $W \in \R^{m \times \ell}$.

However, in the actual algorithm these phases are not separated.  The feature mapping functions are precomputed (oblivious to the data), so the approximate feature vectors are immediately fed into the matrix sketching algorithm, so we never need to fully materialize and store the full $n \times m$ matrix $Z$.  
Also, perhaps unintuitively, we do not sketch the $m$-dimensional column-space of $Z$, rather its $n$-dimensional row-space.  Yet, since the resulting $\ell$-dimensional row-space of $W$ (with $\ell \ll m$) encodes a lower dimensional subspace within $\R^m$, it serves to represent our kernel principal components.    
Pseudocode is provided in Algorithm \ref{alg:skpca}.

\begin{algorithm}[t]
\caption{\label{alg:skpca} \textsc{SKPCA}}
\begin{algorithmic}
\STATE \textbf{Input:} $A \in \R^{n \times d}$ as data points, a shift-invariant kernel function $K$, and $\ell, m \in \Z$
\STATE \textbf{Output:} Feature maps $[f_1,\cdots,f_{m}]$ and their approximate best $\ell$-dim subspace $W$  
\STATE $[f_1,\cdots,f_{m}]$ = \textsc{FeatureMaps}($K, m$) 
\STATE $B \leftarrow 0^{\ell \times m}$
\FOR {$i \in [n]$}
  \STATE $z_i = \sqrt{\frac{2}{m}} [f_1(a_i),\cdots,f_{m}(a_i)]$ 
  \STATE $B \leftarrow z_i$ \hfill \#\text{insert $z_i$ as a row to $B$}
  \IF {$B$ has no zero valued rows}
      \STATE $[Y,\Sigma, W] \leftarrow \svd(B)$ 
      \STATE $B \leftarrow \sqrt{\max\{0,\Sigma^2 - \Sigma_{\ell/2,\ell/2}^2 I_{\ell}\}} \cdot W^T$ 
  \ENDIF
\ENDFOR
\STATE \textbf{Return} $[f_1,\cdots,f_{m}]$ and $W$ \hfill \#$W\in\R^{m\times \ell}$ 
\end{algorithmic}
\end{algorithm}


\paragraph{Approximate feature maps.}
To make the algorithm concrete, we consider the approximate feature maps described in the general framework of Rahimi and Recht~\cite{rahimi2007random}; label this instantiation of the \textsc{FeatureMaps} function as \textsc{Random-Fourier FeaureMaps} (or \RFM). 
This works for positive definite shift-invariant kernels $K(x,y)=K(x-y)$ (e.g.\ Gaussian kernel $K(x,y) = (1/2\pi)^{d/2} \exp(-\|x-y\|^2 /2)$).  It computes a randomized feature map $z:\R^d \to \R^m$ so that $\E[z(x)^T z(y)] = K(x,y)$ for any $x,y \in \R^d$.  
To construct the mapping $z$, they define $m$ functions of the form $f_i(x) = \cos(r_i^Tx + \gamma_i)$, where $r_i \in \R^{d}$ is a sample drawn uniformly at random from the Fourier transform of the kernel function, and $\gamma_i \sim \textsf{Unif}(0,2\pi]$, uniformly at random from the interval $(0,2\pi]$.   
Applying each $f_i$ on a datapoint $x$, gives the $i$th coordinate of $z(x)$ in $\R^{m}$ as $z(x)_i = \sqrt{2/m}f_i(x)$. This implies each coordinate has squared value of $(z(x)_i)^2 \leq 2/m$.    

We consider  $m = O((1/\eps^2) \log n)$ and $\ell = O(1/\eps)$. 

\textsc{\underline{Space}: }
We store the $m$ functions $f_i$, for $i = 1,\ldots,m$; 
since for each function uses a $d$-dimensional vector $r_i$, it takes $O(dm)$ space in total.
We compute feature map $z(x)$ and get a $m$-dimensional row vector $z(a_i)$ for each data point $a_i \in A$, which then is used to update the sketch $B \in \R^{\ell \times m}$ in \FD.
Since we need an additional $O(\ell m)$ for storing $B$ and $W$, the total space usage of Algorithm \ref{alg:skpca} is $O(dm + \ell m) = O((d\log n)/\eps^2 + (\log n)/\eps^3)$. 

\textsc{\underline{Train Time}: }
Applying the feature map takes $O(n\cdot dm)$ time and computing the modified \FD sketch takes $O(n \ell m)$ time, so the training time is $O(n dm + n\ell m) = O(n \log n (d/\eps^2 + 1/\eps^3))$.

\textsc{\underline{Test Time}: }
For a test point $x_{\text{test}}$, we can lift it to $\R^m$ in $O(dm)$ time using $\{f_1, \ldots, f_m\}$, and then use $W$ to project it to $\R^\ell$ in $O(\ell m)$ time. In total it takes $O(dm + \ell m) = O((d+1/\eps)/\eps^2 \cdot \log n)$ time.  

Although $W$ approximates the eigenspace of $A$, it will be useful to analyze the error by also considering all of the data points lifted to $\R^m$ as the $n \times m$ matrix $Z$; and then its projection to $\R^\ell$ as $\tilde{Z} = ZW \in \R^{n \times \ell}$.  
Note $\tilde Z$ would need an additional $O(n \ell)$ to store, and another pass over $A$ (similar for \Nyst and \RFF); we \emph{do not} compute and store $\tilde Z$, only analyze it.

\section{Error Analysis}
In this section, we prove our main error bounds. 
Let $G = \Phi \Phi^T$ be the exact kernel matrix in \RKHS. 
Let $\hat G = Z Z^T$ be an approximate kernel matrix using $Z \in \R^{n\times m}$; it consists of mapping the $n$ points to $\R^m$ using $m$ \RFM.  
In addition, consider $\tilde G = ZWW^TZ^T$ as the kernel matrix which could be constructed from output $W$ of Algorithm \ref{alg:skpca} using \RFM.  We prove that Algorithm \ref{alg:skpca} satisfies the two following error bounds:
\begin{itemize}
\item Spectral error bound: $\|G - \tilde G\|_2 \leq \eps n$
\item Frobenius error bound: $\|G-\tilde G_k\|_F \leq \|G-G_k\|_F + \eps \sqrt{k} n$
\end{itemize}

\subsection{Spectral Error Analysis}\label{sec:spectral_err}
 

The main technical challenge to prove the spectral error bound is that FD bounds are typically for the covariance matrix $W^T W$ not the gram matrix $W W^T$, and thus new ideas are required.  
%
%
%
The breakdown of the proof is as following:

First, in Lemma \ref{lem:HiltoRR}, we show that if $m = O((1/\eps^2)\log (n/\delta))$ then $\|G - \hat G\|_2 \leq \eps n$ with probability at least $1-\delta$.  
Note this is a ``for all'' result which (see Lemma \ref{lem:GtoPhi}) is equivalent to $| \|\Phi^T x\|^2 - \|Z^T x\|^2 | \leq \eps n$ for \emph{all} unit vectors $x$.  If we loosen this to a ``for each'' result, where the above inequality holds for any one unit vector $x$ (say we only want to test with a single vector $x_\text{test}$), then we show in Theorem \ref{thm:phi-z} that this holds with $m = O((1/\eps^2) \log (1/\delta))$.  
This makes partial progress towards an open question~\cite{ANW14} (can $m$ be independent of $n$?) about creating oblivious subspace embeddings for Gaussian kernel features.  

Second, in Lemma \ref{lem:RRtoFD}, we show that applying the modified Frequent Directions step to $Z$ does not asymptotically increase the error. To do so, we first show that spectrum of $Z$ along the directions that \fd fails to capture is small. We prove this for any matrix $A\in \R^{n\times m}$ that is approximated as $B \in \R^{\ell \times m}$ by \fd.

In the following section, we prove the ``for all'' spectral error bound. To do so, we extend the analysis of Lopez \etal~\cite{lopez2014randomized} to show that the Fourier Random Features of Rahimi and Recht~\cite{rahimi2007random} approximate the spectral error with their approximate Gram matrix within $\eps n$ with high probability.  

\subsubsection{``For All'' Bound for RFFMAPS}

In our proof we use the Bernstein inequality on sum of zero-mean random matrices.

\vspace{1.5mm}
Matrix Bernstein Inequality: 
Let $X_1,\cdots,X_d \in \R^{n\times n}$ be independent random matrices such that for all $1 \leq i \leq d$, $E[X_i] = 0$ and $\|X_i\|_2 \leq R$ for a fixed constant $R$. Define variance parameter as $\sigma^2 = \max \{\|\sum_{i=1}^d \E[X_i^T X_i]\|,\|\sum_{i=1}^d \E[X_i X_i^T]\| \}$. Then for all $t \geq 0$, 
$\Pr\left[\Big\|\sum_{i=1}^d X_i\Big\|_2 \geq t \right] \leq 2n \cdot \exp\left(\frac{-t^2}{3\sigma^2+2Rt}\right)$.
Using this inequality, \cite{lopez2014randomized} bounded $\E[\|G-\hat{G}\|_2]$. Here we employ similar ideas to improve this to a bound on $\|G-\hat{G}\|_2$ with high probability.

\begin{lemma}\label{lem:HiltoRR}
For $n$ points, let $G = \Phi \Phi^T \in \R^{n \times n}$ be the exact gram matrix, and let $\hat{G} = ZZ^T \in \R^{n \times n}$ be the approximate kernel matrix using $m = O((1/\eps^2) \log (n/\delta))$ \RFM. 
Then $\|G-\hat{G}\| \leq \eps n$ with probability at least $1-\delta$.
\end{lemma} 
\begin{proof}
Consider $m$ independent random variables $E_i = \frac{1}{m} G - z_i z_i^T$. Note that $\E[E_i] = \frac{1}{m} G - \E[z_i z_i^T] = 0^{n \times n}$ \cite{rahimi2007random}.   Next we can rewrite 
\[
\|E_i\|_2 = \left\|\frac{1}{m} G - z_i z_i^T \right\|_2 = \left\|\frac{1}{m} \E[ZZ^T] - z_i z_i^T\right\|_2
\]
 and thus bound
\begin{align*}
\|E_i\|_2 
&\leq 
\frac{1}{m} \|\E[ZZ^T]\|_2 + \|z_i z_i^T\|_2 
\leq 
\frac{1}{m} \E[\|Z\|_2^2] + \|z_i\|^2  \leq 
\frac{2n}{m} + \frac{2n}{m} = \frac{4n}{m}
\end{align*}
The first inequality is correct because of triangle inequality, and second inequality is achieved using Jensen's inequality on expected values, which states $\|\E[X]\| \leq \E[\|X\|]$ for any random variable $X$.
Last inequality uses the bound on the norm of $z_i$ as $\|z_i\|^2 \leq \frac{2n}{m}$, and therefore $\|Z\|_2^2 \leq \|Z\|_F^2 \leq 2n$.

To bound $\sigma^2$, due to symmetry of matrices $E_i$, simply $\sigma^2 = \|\sum_{i=1}^m \E[E_i^2]\|_2$.
Expanding 
\begin{align*}
\E[E_i^2] 
&=
\E\left[\left(\frac{1}{m} G - z_i z_i^T\right)^2\right] =
\E\left[\frac{G^2}{m^2} + \|z_i\|^2 z_iz_i^T - \frac{1}{m} (z_i z_i^T G + G z_i z_i^T) \right]
\end{align*}
it follows that
\begin{align*}
\E[E_i^2] &\leq \frac{G^2}{m^2} + \frac{2n}{m} \E[z_i z_i^T] - \frac{1}{m}(\E[z_iz_i^T] G + G\;\E[z_i z_i^T]) \\
&= \frac{1}{m^2} (G^2 + 2n G - 2 G^2) = \frac{1}{m^2} (2n G - G^2)
\end{align*}
The first inequality holds by $\|z_i\|^2 \leq 2n/m $, and second inequality is due to $\E[z_i z_i^T] = \frac{1}{m} G$. 
Therefore 
\begin{align*}
\sigma^2 &= \left\|\sum_{i=1}^m \E[E_i^2]\right\|_2 \leq \left\|\frac{1}{m} (2n\; G - G^2) \right\|_2 \\
&\leq \frac{2n}{m} \|G\|_2 + \frac{1}{m} \|G^2\|_2 \leq \frac{2n^2}{m} + \frac{1}{m} \|G\|_2^2 \leq \frac{3n^2}{m}
\end{align*}
the second inequality is by triangle inequality, and the last inequality by $\|G\|_2 \leq \Tr{(G)} = n$.
Setting $M = \sum_{i=1}^m E_i = \sum_{i=1}^m (\frac{1}{m}G - z_{:,i}z_{:,i}^T) = G - \hat{G}$ and using Bernstein inequality with $t = \eps n$ we obtain 
\begin{align*}
\Pr \left[\|G-\hat G\|_2 \geq \eps n \right] &\leq 2n \exp \left(\frac{-(\eps n)^2}{3(\frac{3n^2}{m}) + 2 (\frac{4n}{m}) \eps n} \right) = 2n \exp \left( \frac{-\eps^2 m}{9 + 8\eps} \right) \leq \delta
\end{align*}
Solving for $m$ we get $m \geq \frac{9 + 8\eps}{\eps^2} \log(2n/\delta)$, so with probability at least $1-\delta$ for $m = O(\frac{1}{\eps^2} \log(n/\delta))$, then $\|G-\hat G\|_2 \leq \eps n$.
\end{proof}


Next, we prove the ``for each'' error bound. 
\subsubsection{``For Each'' Bound for \RFM}
\label{app:RFF-whp-each}

Here we bound $\left|\|\Phi^T x\|^2 - \|Z^Tx\|^2\right|$, where $\Phi$ and $Z$ are mappings of data to \RKHS by \RFM, respectively and $x$ is a \textit{fixed} unit vector in $\R^n$.

Note that Lemma \ref{lem:HiltoRR} essentially already gave a stronger proof, where using $m = O((1/\eps^2) \log (n/\delta))$ the bound $\|G - \hat G\|_2 \leq \eps n$ holds along all directions (which makes progress towards addressing an open question of constructing oblivious subspace embeddings for Gaussian kernel features spaces, in \cite{ANW14}).  
The advantage of this proof is that the bound on $m$ will be independent of $n$.  
Unfortunately, in this proof, going from the ``for each'' bound to the stronger ``for all'' bound would seem to require a net of size $2^{O(n)}$ and a union bound resulting in a worse ``for all'' bound with $m = O(n/\eps^2)$.  

On the other hand, main objective of \textsc{Test time} procedure, which is mapping a single data point to the $D$-dimensional or $k$-dimensional kernel space is already interesting for what the error is expected to be for a single vector $x$.  This scenario corresponds to the ``for each'' setting that we will prove in this section.  

In our proof, we use a variant of Chernoff-Hoeffding inequality, stated next.
Consider a set of $r$ independent random variables $\{X_1,\cdots, X_r\}$ where $0 \leq X_i \leq \Delta$. Let $M = \sum_{i=1}^r X_i$, then for any $\alpha \in (0,1/2)$,
$\Pr\left[|M-\E[M]| > \alpha \right]\leq 2\exp\left(\frac{-2\alpha^2}{r \Delta^2}\right)$.

For this proof we are more careful with notation about rows and column vectors.  Now matrix $Z \in \R^{n \times m}$ can be written as a set rows $[z_{1,:}; z_{2,:}; \ldots, z_{n,:}]$ where each $z_{i,:}$ is a vector of length $m$ or a set of columns $[z_{:,1}, z_{:,2}, \ldots, z_{:,d}]$, where each $z_{:,j}$ is a vector of length $n$. We denote the $(i,j)$-th entry of this matrix as $z_{i,j}$.

\begin{theorem}
\label{thm:phi-z}
For $n$ points in any arbitrary dimension and a shift-invariant kernel, let $G = \Phi \Phi^T \in \R^{n \times n}$ be the exact gram matrix, and $\hat{G} = ZZ^T \in \R^{n \times n}$ be the approximate kernel matrix using $m = O((1/\eps^2) \log(1/\delta))$ \RFM.  
Then for any fixed unit vector $x \in \R^n$, it holds that $\left|\|\Phi^Tx\|^2 - \|Z^Tx\|^2\right| \leq \eps n$ with probability at least $1-\delta$. 
\end{theorem}
\begin{proof}
Note $\R^n$ is not the dimension of data.
Consider any unit vector $x \in \R^n$.  
Define $m$ independent random variables $\{X_i=\langle {z}_{:,i}, x\rangle^2\}_{i=1}^m$. We can bound each $X_i$ as $0 \leq X_i \leq \|{z}_{:,i}\|^2 \leq 2n/m$ therefore $\Delta = 2n/m$ for all $X_i$s. 
Setting $M = \sum_{i=1}^{m} X_i = \|Z^Tx\|^2$, we observe 
\begin{align*}
&\E[M] 
= 
\sum_{i=1}^{m} \E\left[\langle {z}_{:,i}, x\rangle^2\right] = 
\sum_{i=1}^{m} \E\left[(\sum_{j=1}^n z_{ji} \; x_j)^2\right] \\
&=\sum_{i=1}^{m} \E\left[\sum_{j=1}^n (z_{ji} \; x_j)^2 + 2\sum_{j=1}^n \sum_{k>j}^n z_{ji} \; z_{ki} \; x_j \; x_k\right] 
\\&= 
\sum_{j=1}^n x_j^2 \;\E\left[\sum_{i=1}^{m} z_{ji}^2\right] + 2\sum_{j=1}^n \sum_{k>j}^n x_j\;x_k \;\E\left[ \sum_{i=1}^{m} z_{ji} \; z_{ki}\right] \\
&= 
\sum_{j=1}^n x_j^2 \;\E\left[\langle z_{j,:}, z_{j,:}\rangle \right] + 2\sum_{j=1}^n \sum_{k>j}^n x_j\;x_k \;\E\left[\langle z_{j,:}, z_{k,:}\rangle \right] \\
&=\sum_{j=1}^n x_j^2 \;\langle \phi_{j,:}\;,\phi_{j,:}\rangle + 2\sum_{j=1}^n \sum_{k>j}^n x_j\;x_k \;\langle \phi_{j,:}, \phi_{k,:}\rangle  \\
&= 
\sum_{j=1}^n x_j^2 \sum_{i=1}^{D} \phi_{ji}^2 + 2\sum_{j=1}^n \sum_{k>j}^n x_j \; x_k \sum_{i=1}^{D} \phi_{ji} \; \phi_{ki} 
\\&= 
\sum_{i=1}^{D} \left(\sum_{j=1}^n x_j^2 \; \phi_{ji}^2 + 2\sum_{j=1}^n \sum_{k>j}^n x_j \; x_k \; \phi_{ji} \; \phi_{ki} \right) \\
&= 
\sum_{i=1}^{D} \langle \phi_{:,i}, x\rangle^2 
= 
\|\Phi^T x\|^2
\end{align*}
Since $x$ is a fixed unit vector, it is pulled out of all expectations.
Using the Chernoff-Hoeffding bound and setting $\alpha=\eps n$ yields
$\Pr\left[|\|\Phi^Tx\|^2 - \|Z^Tx\|^2| > \eps n\right] \leq 2 \exp\left(\frac{-2 (\eps n)^2}{m (2n/m)^2}\right) = 2\exp\left(-2\eps^2 m\right) \leq \delta$.  
Then we solve for $m = (1/(2\eps^2)) \ln(2/\delta)$ in the last inequality.
\end{proof}

Next, we prove the last piece which shows applying the modified Frequent Directions step to $Z$ does not asymptotically increase the error. 

\subsubsection{Spectral Error Bound for Algorithm \ref{alg:skpca}}
We first show that spectrum of $Z$ along the directions that \fd fails to captures is small. We prove this for any $n\times m$ matrix $A$ that is approximated as $B \in \R^{\ell \times m}$ by \fd.
\begin{lemma} \label{lem:FD-left-null} 
Consider an $A\in\R^{n \times m}$ matrix with $m \leq n$, and let $B$ be an $\ell \times m$ matrix resulting from running Frequent Directions on $A$ with $\ell$ rows. For any unit vector $y \in \R^n$ with $\|y^T A B^\dagger B\| = 0$, it holds that $\|y^T A\|^2 \leq \|A - A_k\|_F^2/(\ell-k)$, for all $k \leq \ell$, including $k=0$ where $A-A_k = A$.  
\end{lemma}
\begin{proof}
Let $[U,S,V] = \svd(A)$ be the \svd\;  of $A$.  
Consider any unit vector $y\in\R^n$ that lies in the column space of $A$ and the null space of $B$, that is $\|y^T A B^\dagger B\| = 0$ and $\|y^T AA^\dagger\| = \|y\| = 1$. 
Since $U = [u_1, u_2, \ldots, u_n]$ provides an orthonormal basis for $\R^n$, we can write 
\[
y = \sum_{i=1}^n \alpha_i u_i  \;\;\;\; \text{such that }\;\;\;\; \alpha_i = \langle y, u_i\rangle , \;\;\;\; \sum_{i=1}^n \alpha_i^2 = 1
\] 
Since $1 = \sum_{i=1}^n \alpha_i^2 = \|y\| = \|y^T AA^\dagger\| = \|y^T U_m U_m^T\| = \sum_{i=1}^m \alpha_i^2$, therefore $\alpha_i = 0$ for $i>m$.  
Moreover $ \|y^T A\|^2 = \sum_{i=1}^m s_i^2 \langle y, u_i \rangle^2 = \sum_{i=1}^m s_i^2 \alpha_i^2$.
This implies there exists a \textit{unit} vector $x = \sum_{i=1}^m \alpha_i v_i \in \R^m$ with $\alpha_i = \langle x,v_i\rangle = \langle y,u_i\rangle$ for $i=1,\cdots,m$ such that $\|y^T A\| = \|A x\|$ and importantly $\|Bx\| = 0$, which we will prove shortly.  

Then, due to the Frequent Directions bound~\cite{gp14}, for any unit vector $\bar x\in\R^m$, $\|A\bar x\|^2 - \|B \bar x\|^2 \leq \|A-A_k\|_F^2/(\ell-k)$, and for our particular choice of $x$ with $\|Bx\|=0$, we obtain $\|y^T A\| = \|Ax\|^2 \leq \|A-A_k\|_F^2/(\ell-k)$, as desired.

Now to see that $\|Bx\| = 0$, we will assume that $\|Bx\| > 0$ and prove a contradiction.  Since $\|Bx\| > 0$, then $x$ is not in the null space of $B$, and $\|\pi_B(x)\| > 0$ for any unit vector $x$.  Let $\Sigma = \diag(\sigma_1, \ldots, \sigma_\ell)$, assuming $\sigma_1 \geq \sigma_2 \geq \ldots \geq \sigma_\ell > 0$, are the singular values of $B$, and $W = [w_1, \ldots, w_\ell] \in \mathbb{R}^{m x \ell}$ are its right singular vectors.  Then $\|Bx\| = \|\Sigma W x\|$ and if $\|\pi_B(x)\| > 0$, then setting $\bar{\Sigma} = \diag(1, 1, \ldots, 1)$ and $\bar{B} = \bar{\Sigma}W = W_\ell$ to remove the scaling from $B$, we have $\|\pi_{\bar B}(x)\| > 0$.  
Similarly, if $\|y^T U S V^T B^\dagger B\| = \|y^T \pi_B(A)\| = 0$, then setting $\bar{S} = \diag(1,\ldots, 1)$ and $\bar{A} = U \bar{S} V^T$ to remove scale from $A$, we have $\|y^T \pi_B(\bar{A})\| = 0$.  
Hence 
\begin{align*}
0 
&< 
\|\pi_{\bar{B}}(x)\| 
= 
\|x B^\dagger B\| 
= 
\|x W_\ell W_\ell^T\|
=
\| \sum_{j=1}^\ell \langle x, w_j \rangle\| = \|\sum_{j=1}^\ell  \sum_{i=1}^m  \alpha_i \langle v_i, w_j\rangle \|
\end{align*}
and 
\begin{align*}
0 
&= 
\|y^T \pi_B(\bar{A})\| 
=
\|\sum_{i=1}^m \langle y, u_i \rangle v_i^T W_\ell W_\ell^T\|
=
\|\sum_{i=1}^m \alpha_i v_i^T W_\ell\| = \|\sum_{j=1}^\ell  \sum_{i=1}^m  \alpha_i \langle v_i, w_j\rangle\|.  
\end{align*}
Since last terms of each line match, we have a contradiction, and hence $\|Bx\| = 0$.
\end{proof}

\begin{lemma}\label{lem:RRtoFD}
Let $\tilde Z = ZW$, and $\tilde G = \tilde Z \tilde Z^T = ZWW^TZ^T$ be the corresponding gram matrix from $Z \in \R^{n \times m}$ and $W \in \R^{m \times \ell}$ constructed via Algorithm \ref{alg:skpca} with $\ell = 2/\eps$. Comparing to $\hat{G} = ZZ^T$, then $\|\tilde G-\hat{G}\|_2 \leq \eps n$.
\end{lemma}

\begin{proof}
Consider any unit vector $y\in \R^n$, and note that $y^TZ = [y^TZ]_W + [y^TZ]_{\perp W}$ where $[y^TZ]_W = y^TZWW^T$ lies on the column space spanned by $W$, and $[y^TZ]_{\perp W} = y^TZ(I-WW^T)$ is in the null space of $W$. 
Then first off $\|y^TZ\|^2 = \|[y^TZ]_W\|^2 + \|[y^TZ]_{\perp W}\|^2$ since two components are perpendicular to each other.   
Second $[y^TZ]_W W = y^TZ WW^T W = y^TZW$ and $[y^TZ]_{\perp W} W = y^TZ (I-WW^T)W = y^TZ (W-W) = 0$. Knowing these two we can say
\begin{align*}
&y^T (ZZ^T - \tilde Z \tilde Z^T) y \\
&= (y^TZ)(y^TZ)^T - (y^TZ)WW^T(y^TZ)^T \\
&= \|y^TZ\|^2 - \left([y^TZ]_W + [y^TZ]_{\perp W}\right) WW^T \left([y^TZ]_W + [y^TZ]_{\perp W}\right)^T \\
&= \|y^TZ\|^2 - (y^TZW) (y^TZW)^T \\
&=  \left(\|[y^TZ]_W\|^2 + \|[y^TZ]_{\perp W}\|^2\right) - \|y^T ZW\|^2 \\
&= \|[y^TZ]_{\perp W}\|^2.  
\end{align*}
The last inequality holds because $\|y^T ZW\| = \|y^T ZWW^T\| = \|[y^T Z]_W\|$ as $W$ is an orthonormal matrix.

To show $\|[y^TZ]_{\perp W}\| \leq \eps \|Z\|_F^2$, consider vector $v = y^TZ(I-WW^T)Z^\dagger$ and let $y^*= v /\|v\|$. Clearly $y^*$ satisfies requirement of Lemma \ref{lem:FD-left-null} as it is a unit vector in $\R^n$ and $\|y^* Z WW^T\| = 0$ as
\begin{align*}
\|y^* Z WW^T\| &= \|y^TZ(I-WW^T)Z^\dagger Z WW^T\| / \|v\| \\
&= \|y^TZ(I-WW^T)WW^T\| / \|v\| = 0.   
\end{align*}
Therefore it satisfies $\|y^* Z\| \leq \|Z-Z_k\|_F^2/(\ell - k)$. Since $\|Z\|_F^2 \leq 2n$ for $k =0$ and $\ell = 2/\eps$, we obtain
\begin{align*}
\|[y^TZ]_{\perp W}\|^2 &= \|y^TZ(I-WW^T)\| \\
&= \|y^TZ(I-WW^T)Z^\dagger Z\|^2 = \|y^* Z\|^2 \|v\|^2 \\
&\leq \|Z\|_F^2 \|v\|^2 /\ell \leq \eps n \|v\|^2.    
\end{align*}
It is left to show that $\|v\| \leq 1$. For that, note $\pi_{ZW}\left(\pi_Z (y^T)\right) =  \pi_{ZW} (y^TZZ^\dagger) = y^T ZZ^\dagger (ZW) (ZW)^\dagger = y^T ZW (ZW)^\dagger = \pi_{ZW}(y^T) $.  The finally we obtain 
\begin{align*}
\|v\|^2 &= \|y^T Z(I-WW^T)Z^\dagger\|^2 \\
&= \|y^T ZZ^\dagger - y^T ZWW^TZ^\dagger\|^2 \\
&= \|\pi_Z(y^T) - \pi_{ZW}(y^T)\|^2 \\
& = \|\pi_Z(y^T) - \pi_{ZW}(\pi_Z(y^T))\|^2 \\
&\leq \|\pi_Z(y^T)\|^2 \leq \|y\|^2 = 1.  \hspace{.4in} \qedhere
\end{align*}
%
\end{proof}

Combining Lemmas \ref{lem:HiltoRR} and \ref{lem:RRtoFD} and using triangle inequality, we get our main result.  

\begin{theorem} \label{thm:main}
Let $G = \Phi \Phi^T$ be the exact kernel matrix over $n$ points.  
Let $\tilde G = ZW^TWZ^T$ be the result of $Z$ from $m = O((1/\eps^2) \log (n/\delta))$ \RFM and $W$ from running Algorithm \ref{alg:skpca} with $\ell = 4/\eps$.  
Then with probability at least $1-\delta$,  we have $\|G-\tilde G\|_2 \leq \eps n$.
\end{theorem}


Next, in the following section we prove our frobenius error bound for Algorithm \ref{alg:skpca}.
\subsection{Frobenius Error Analysis}
\label{sec:gen-err}

Let the true gram matrix be $G = \Phi \Phi^T$, and consider $G' = YY^T$, for any $Y$ including when $Y = ZW$.
First we write the bound in terms of $\Phi$ and $Y$.  
\begin{lemma}
\label{lem:GtoPhi}
$\|G - G'\|_2 = \max\limits_{\|x\|=1}  | \|\Phi^T x\|^2 - \|Y^Tx\|^2|$.
\end{lemma} 

\begin{proof}
Recall we can rewrite spectral norm as $\|G - G'\|_2 = \max\limits_{\|x\|=1} |x^TGx - x^T G'x| = \max\limits_{\|x\|=1} |x^T \Phi \Phi^T x - x^TY Y^Tx| = \max\limits_{\|x\|=1} |\|\Phi^T x\|^2 - \|Y^Tx\|^2|$.
First line follows by definition of top eigenvalue of a symmetric matrix, and last line is true because $\|y\|^2 = y^T y$ for any vector $y$.\qedhere
\end{proof}

Thus if $\|G - G'\|_2 \leq \eps n$ where $G' = Y Y^T$ could be reconstructed by any of the algorithms we consider, then it implies $\max_{\|x\|=1} |\|\Phi^T x\|^2 - \|Y^Tx\|^2| \leq \eps n$.  
We can now generalize the spectral norm bound to Frobenius norm.  Let $G-G' = U\Lambda U^T$ be the eigen decomposition of $G-G'$. Recall that one can write each eigenvalue as $\Lambda_{i,i} = u_i^T (G - G') u_i$, and the definition of the Frobenius norm implies 
$\|G - G'\|_F^2 = \sum\limits_{i=1}^n \Lambda_{i,i}^2$
Hence
\begin{align*}
\|G - G'\|_F^2 &=\sum_{i=1}^n (u_i^T (G-G') u_i)^2 = \sum_{i=1}^n (\|\Phi^T u_i\|^2 - \|Y^T u_i\|^2)^2 \leq \sum_{i=1}^n (\eps n)^2 \leq \eps^2 n^3 
\end{align*}
Therefore $\|G-G'\|_F \leq \eps n^{1.5}$.
We can also show a more interesting bound by considering $G_k$ and $G'_k$, the best rank $k$ approximations of $G$ and $G'$ respectively.  
\begin{lemma}
Given that $\|G - G'\|_2 \leq \eps n$ we can bound 
$\|G - G'_k\|_F \leq \|G - G_k\|_F + \eps \sqrt{k} n$.
\end{lemma}
\begin{proof}
Let $[u_1, \ldots, u_n]$ and $[v_1, \ldots, v_n]$ be eigenvectors of $G$ and $G - G'$, respectively. Then
\begin{align*}
\|G - G'_k\|_F^2 
&= 
\sum_{i=1}^k (v_i^T (G-G'_k) v_i)^2 \\
& \hspace{2mm}+ \sum_{i=k+1}^n (v_i^T (G-G'_k) v_i)^2 
\\& \leq 
\sum_{i=1}^k (v_i^T (G-G'_k) v_i)^2 + \sum_{i=k+1}^n (v_i^T G v_i)^2 
\\& \leq 
\sum_{i=1}^k (v_i^T (\Phi\Phi^T - YY^T) v_i)^2 + \sum_{i=k+1}^n (u_i^T G u_i)^2 
\\ &= 
\sum_{i=1}^k (\|\Phi^T v_i\|^2 - \|Y^T v_i\|^2)^2 + \|G - G_k\|_F^2 
\\& \leq 
k (\eps n)^2 + \|G - G_k\|_F^2.
\end{align*}
The second transition is true because $G'$ is positive semidefinite, therefore $v_i^T (G-G'_k) v_i \leq v_i^T G v_i$, and third transition holds because if $u_i$ is $i$th eigenvector of $G$ then, $u_i^T G u_i \geq v_i^T G v_i$ where $v_i$ is $i$th eigenvector of $G - G'$.
Taking square root yields 
\[
\|G - G'_k\|_F^2 \leq \sqrt{\|G - G_k'\|_F^2 + (\eps n)^2 k} \leq \|G - G_k'\|_F + \eps n \sqrt{k}. \qedhere
\] 
\end{proof}

Thus we can get error bounds for the best rank-$k$ approximation of the data in RKHS that depends on ``tail'' $\|G-G_k\|_F$ which is typically small.  We can also make the second term $\eps n \sqrt k$ equal to $\eps' n$ by using a value of $\eps = \eps' / \sqrt k$ in the previously described algorithms.

\begin{figure*}[t!]

\rotatebox{90}{\tiny \hspace{10mm}\textsf{Kernel Frobenius Error}} 
\includegraphics[width=\figsize]{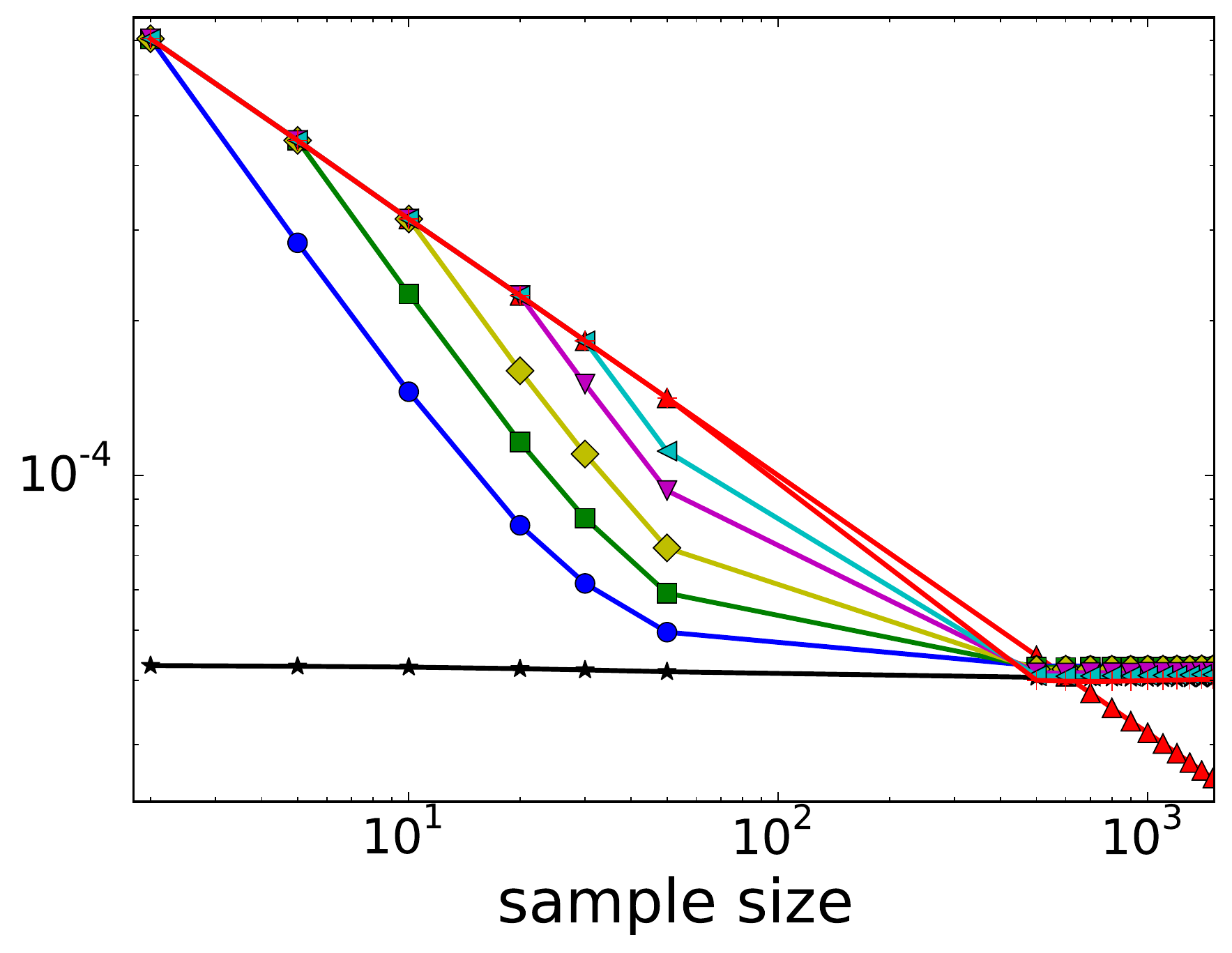}
%
\rotatebox{90}{\tiny \hspace{10mm}\textsf{Kernel Spectral Error}}
\includegraphics[width=\figsize]{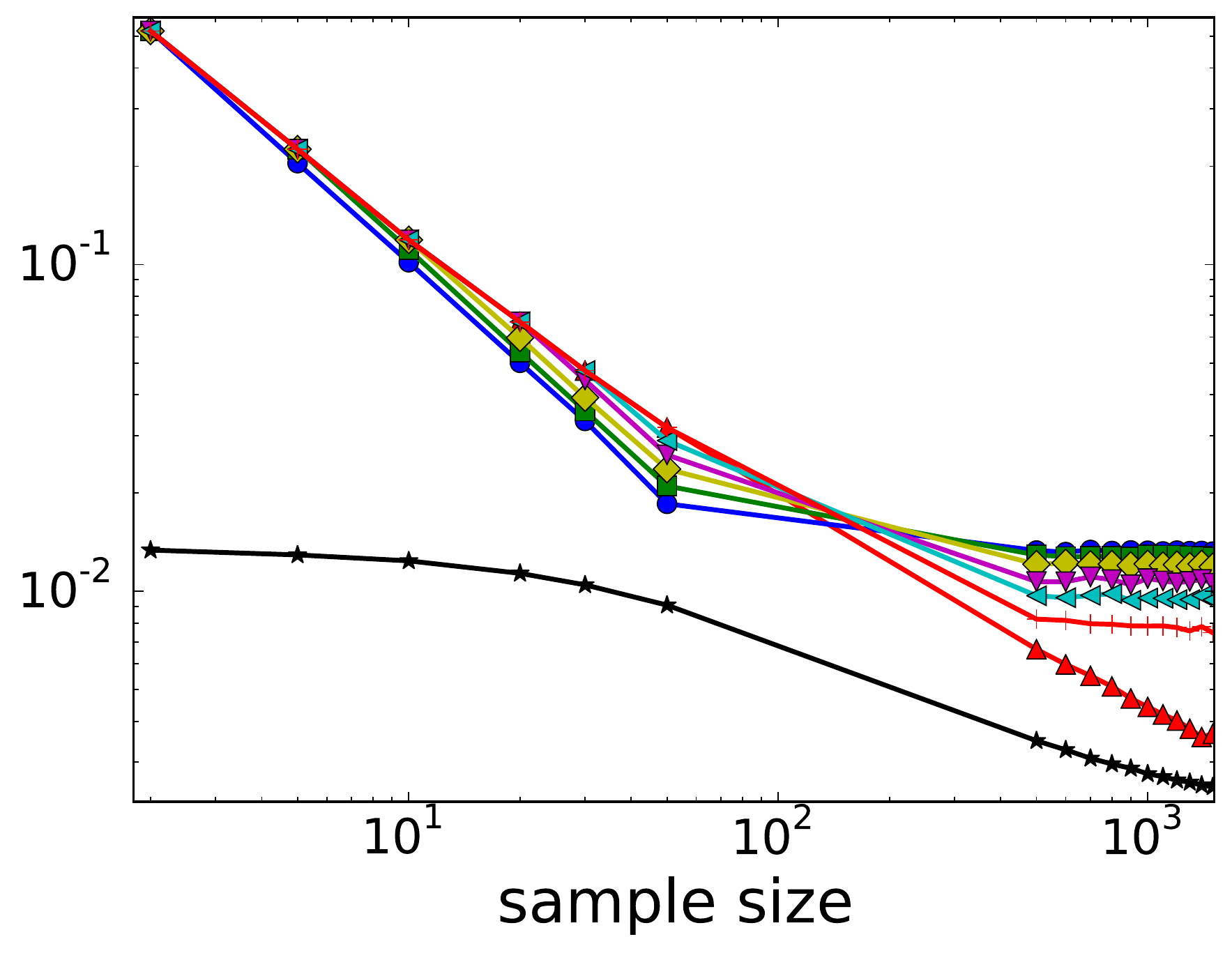}
%
\rotatebox{90}{\tiny \hspace{14mm}\textsf{Train time} (sec)}
\includegraphics[width=\figsize]{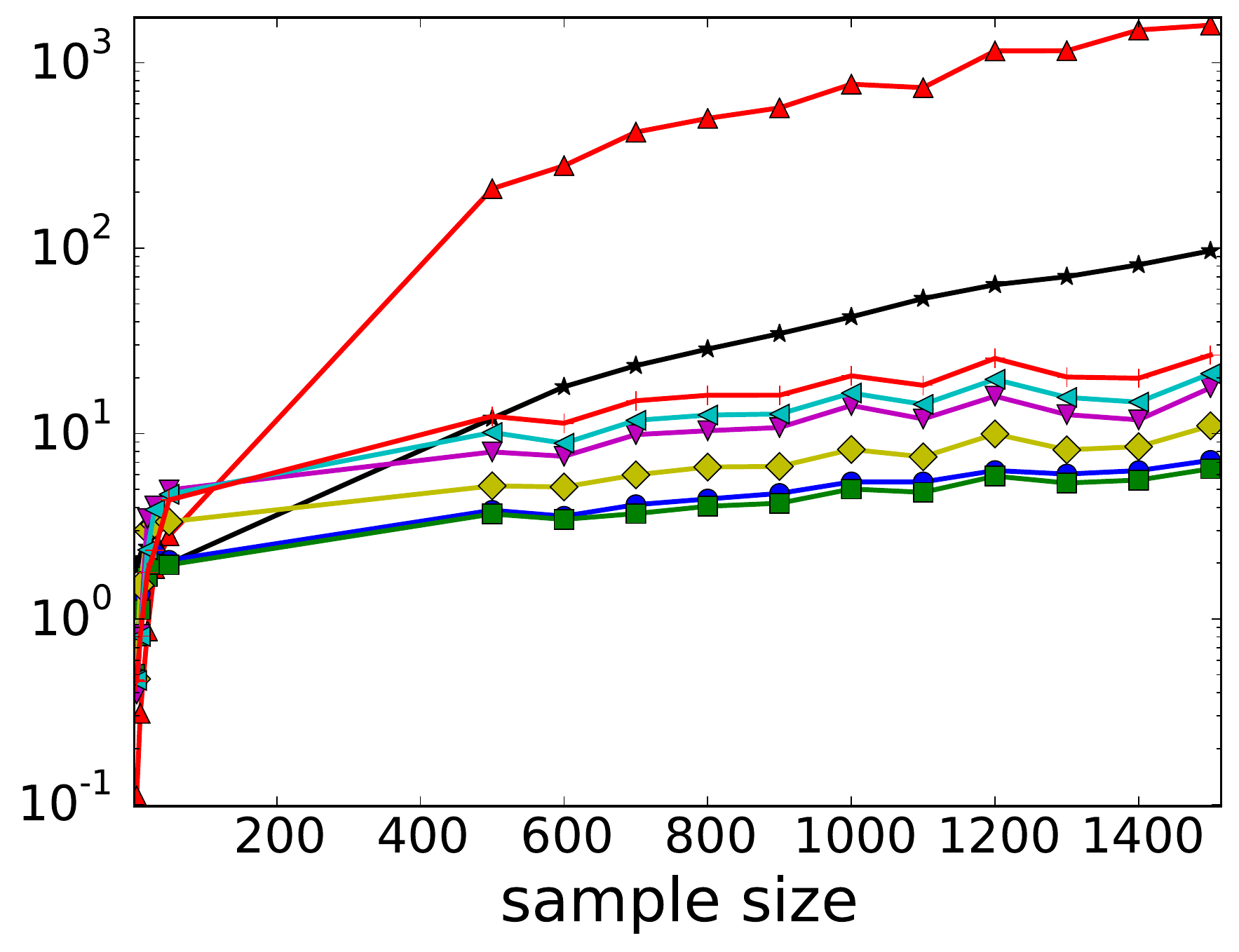}

\rotatebox{90}{\tiny \hspace{10mm}\textsf{Kernel Frobenius Error}}
\includegraphics[width=\figsize]{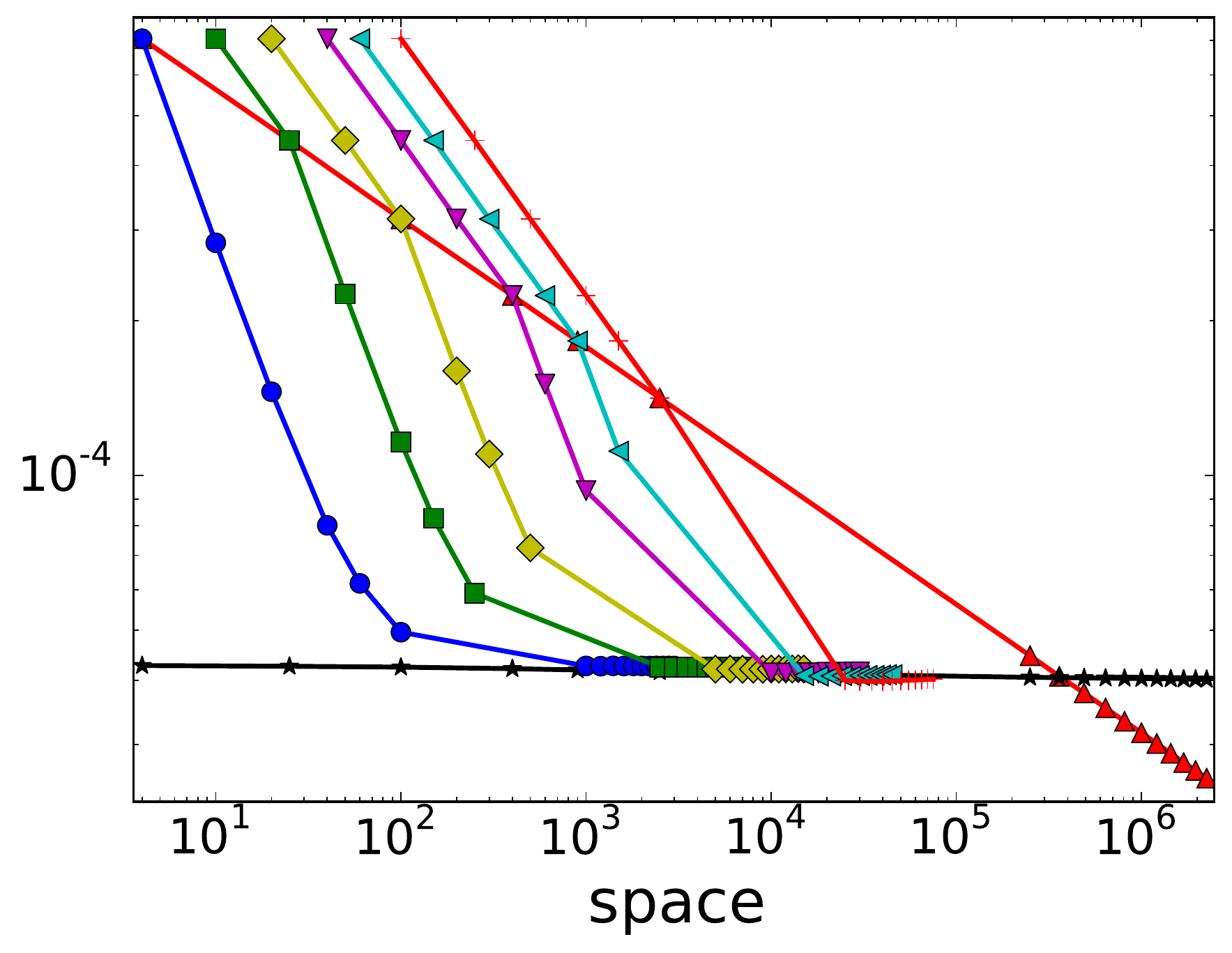}
%
\rotatebox{90}{\tiny \hspace{10mm}\textsf{Kernel Spectral Error}}
\includegraphics[width=\figsize]{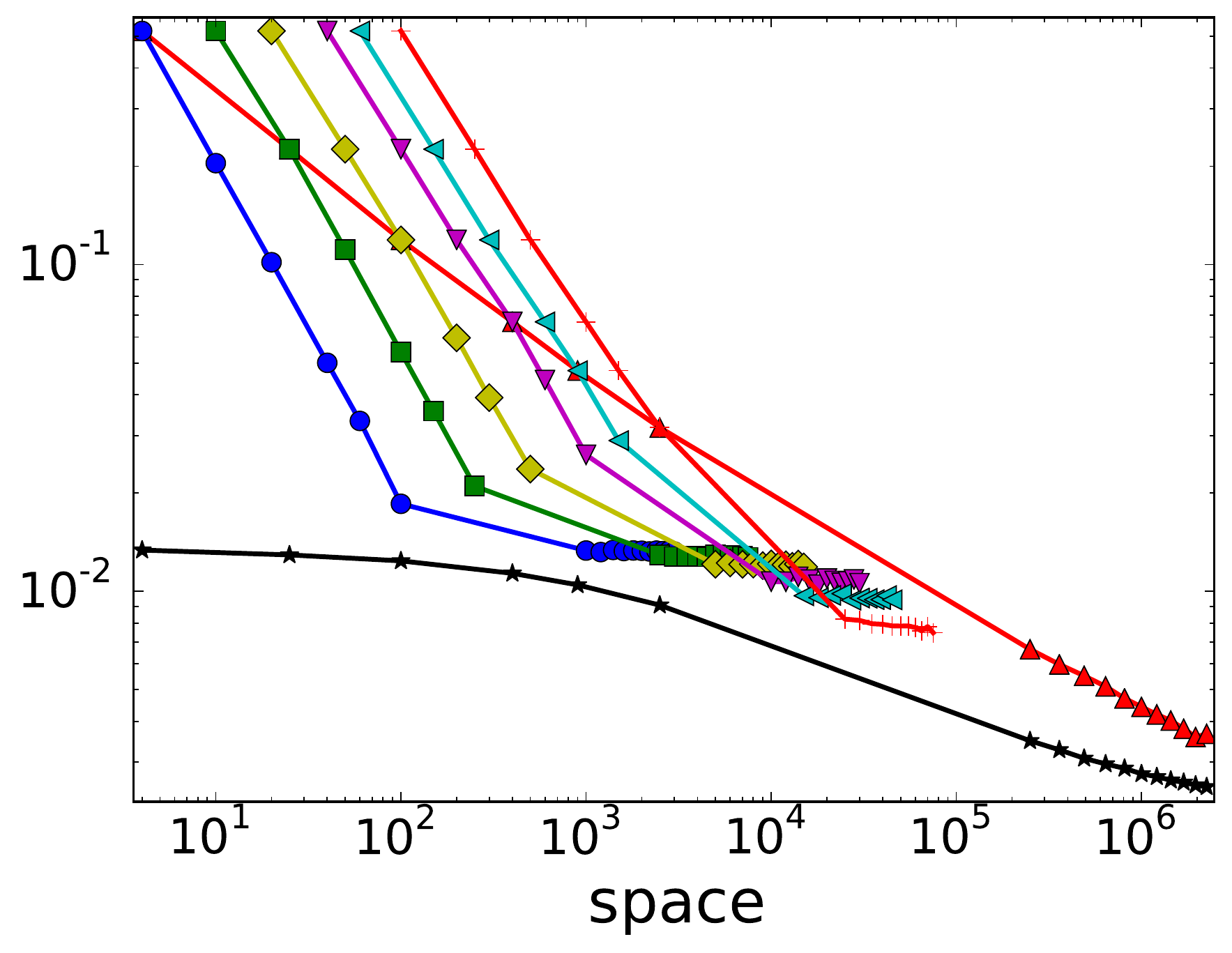}
%
\rotatebox{90}{\tiny \hspace{14mm}\textsf{Test time} (sec)}
\includegraphics[width=\figsize]{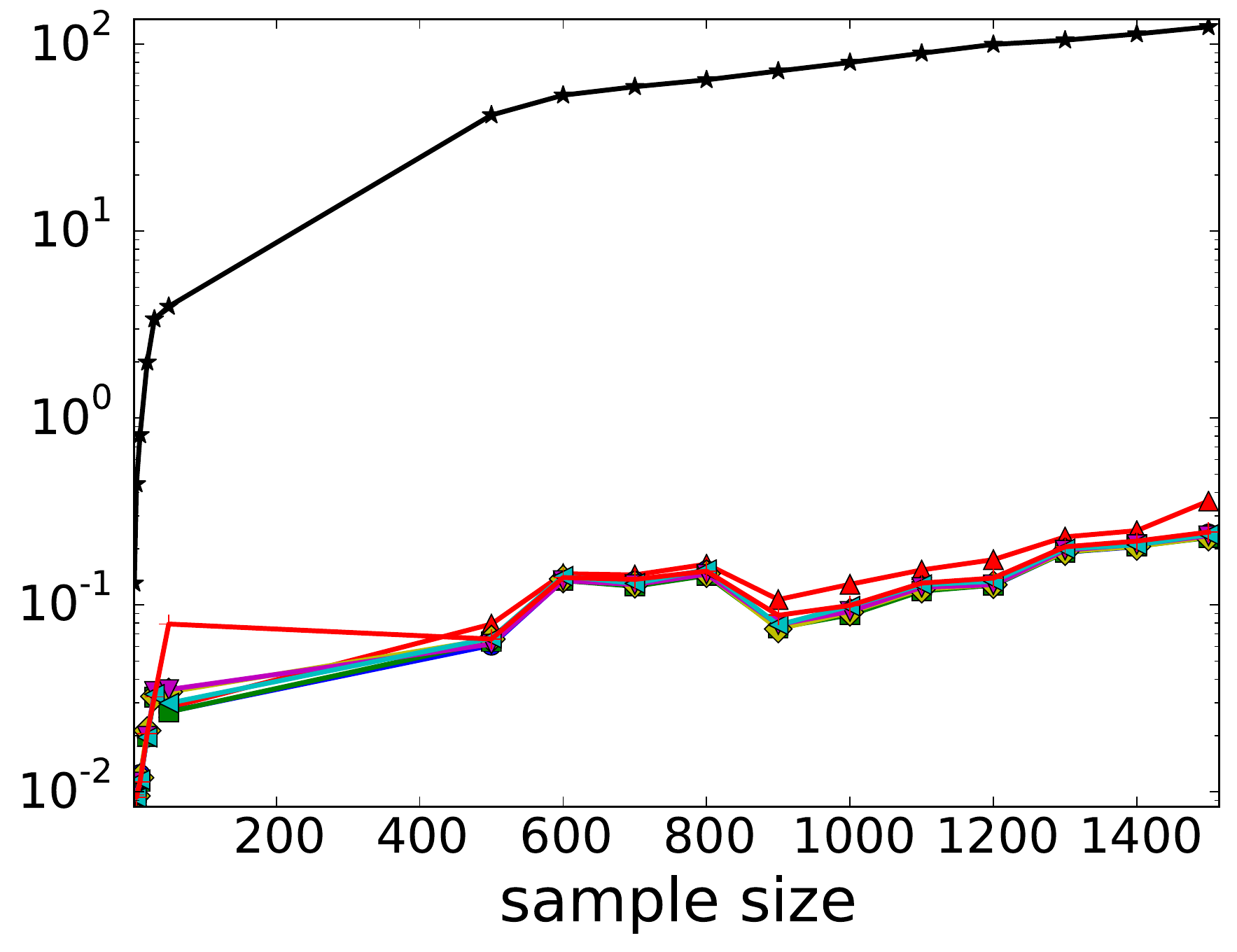}
%

\includegraphics[width=3.3\figsize]{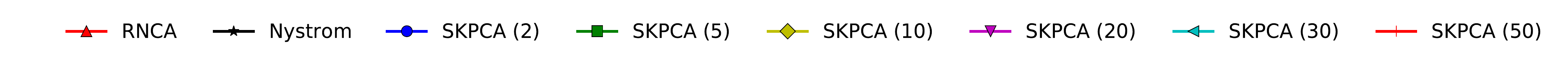}


\caption{{\small \textsc{RandomNoisy} dataset showing Error and Test and Train time versus Sample Size ($m$ or $c$) and Space.}} 


\label{fig:random_noisy}
\end{figure*}
\section{Experiments}
\label{sec:exp}

\label{sec:setup}

We measure the \textsc{Space}, \textsc{Train time}, and \textsc{Test time} of our SKPCA algorithm with $\ell$ taking values $\{2,5,10,20,30,50\}$.  We use spectral and Forbenious error measures and compare against the \Nyst and the RNCA approach using \RFM. 
All methods are implemented in Julia, run on an OpenSUSE 13.1 machine with 80 Intel(R) Xeon(R) 2.40GHz CPU and 750 GB of RAM.

\paragraph{Data sets.}
We run experiments on several real (\textsc{CPU, Adult, Forest}) and synthetic (\textsc{RandomNoisy}) data sets. 
Each data set is an $n \times d$ matrix $A$ (\textsc{CPU} is $7373 \times 21$, \textsc{Forest} is $523910 \times 54$, \textsc{Adult} is $33561 \times 123$, and \textsc{RandomNoisy} is $20000 \times 1000$) with $n$ data points and $d$ attributes.  
For each data set, a random subset is removed as the test set of size $1000$, except \textsc{CPU} where the test set size is $800$.
We generate the \textsc{RandomNoisy} synthetic dataset using the approach by \cite{l13}.  We create $A = SDU + F/\zeta$, where $SDU$ is an $s$-dimensional signal matrix (for $s<d$) and $F/\zeta$ is (full) $d$-dimensional noise with $\zeta$ controlling the signal to noise ratio.  Each entry $F_{i,j}$ of $F$ is generated i.i.d. from a normal distribution $N(0,1)$, and we use $\zeta=10$.    
For the signal matrix, $S \in \R^{n \times s}$ again we generate each $S_{i,j} \sim N(0,1)$ i.i.d;  $D$ is diagonal with entries $D_{i,i} = 1-(i-1)/d$ linearly decreasing; and $U \in \R^{s \times d}$ is just a random rotation.  We set $s=50$.

\begin{figure*}[t!]

\rotatebox{90}{\tiny \hspace{10mm}\textsf{Kernel Frobenius Error}} 
\includegraphics[width=\figsize]{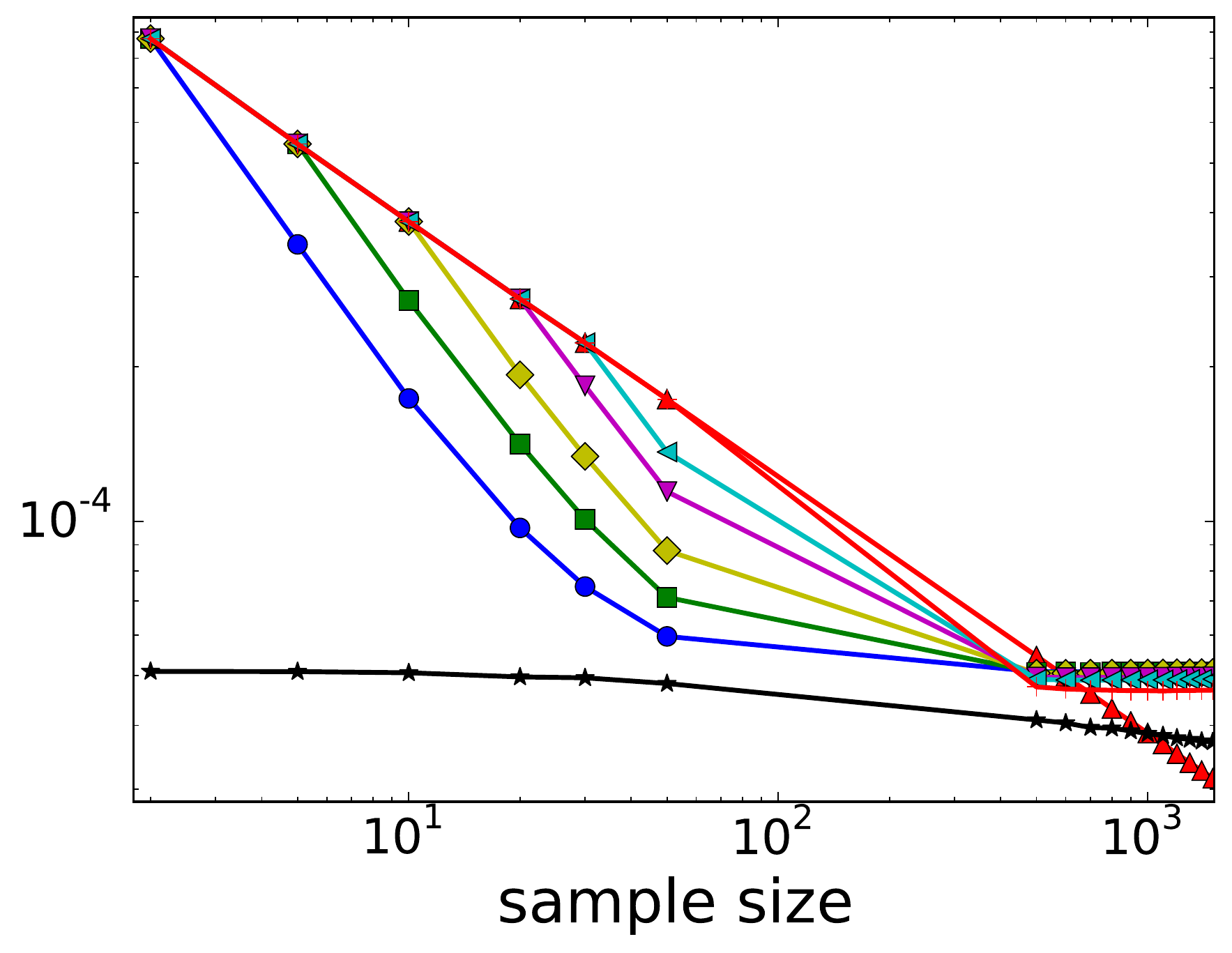}
%
\rotatebox{90}{\tiny \hspace{10mm}\textsf{Kernel Spectral Error}}
\includegraphics[width=\figsize]{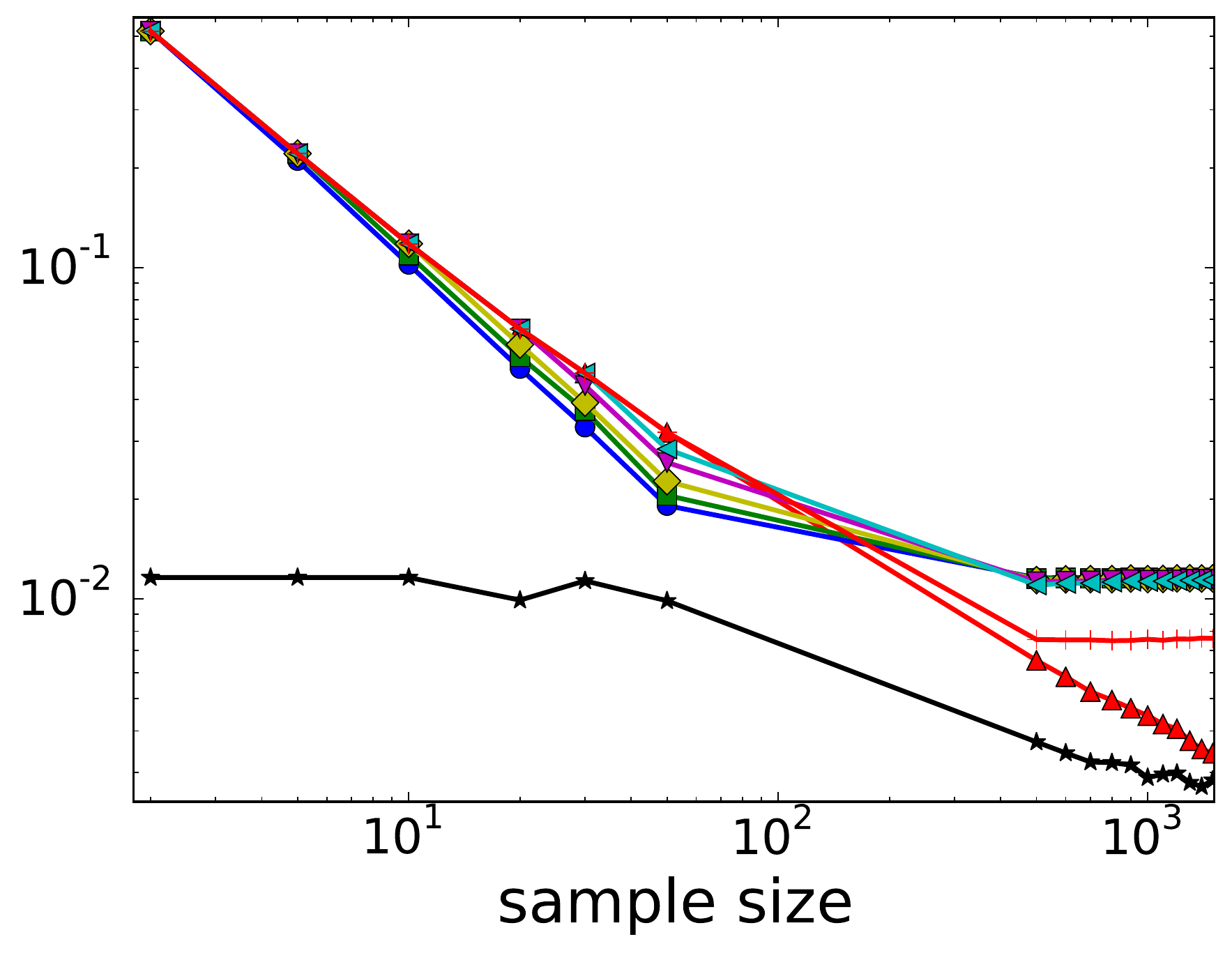}
%
\rotatebox{90}{\tiny \hspace{14mm}\textsf{Train time} (sec)}
\includegraphics[width=\figsize]{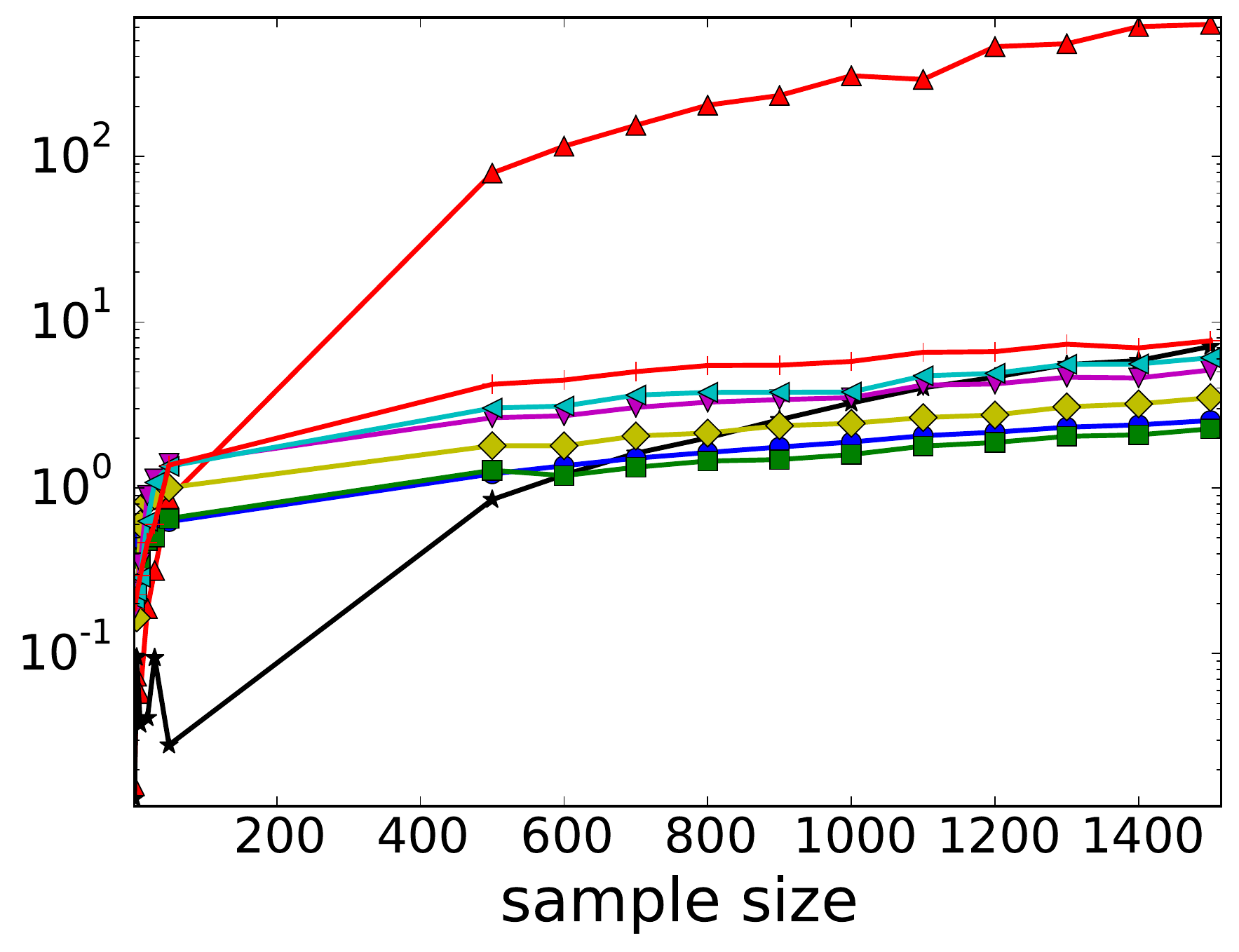}

\rotatebox{90}{\tiny \hspace{10mm}\textsf{Kernel Frobenius Error}}
\includegraphics[width=\figsize]{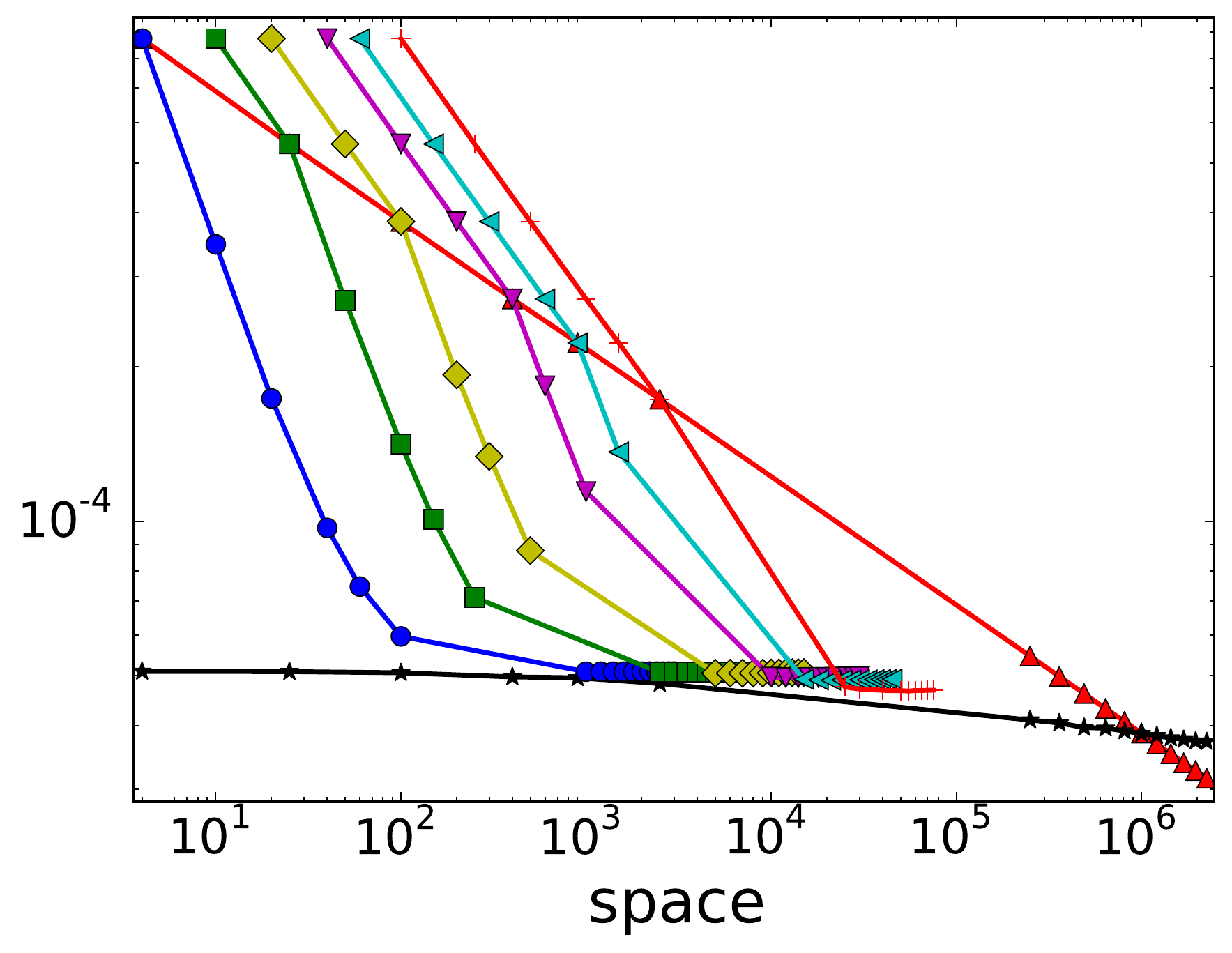}
%
\rotatebox{90}{\tiny \hspace{10mm}\textsf{Kernel Spectral Error}}
\includegraphics[width=\figsize]{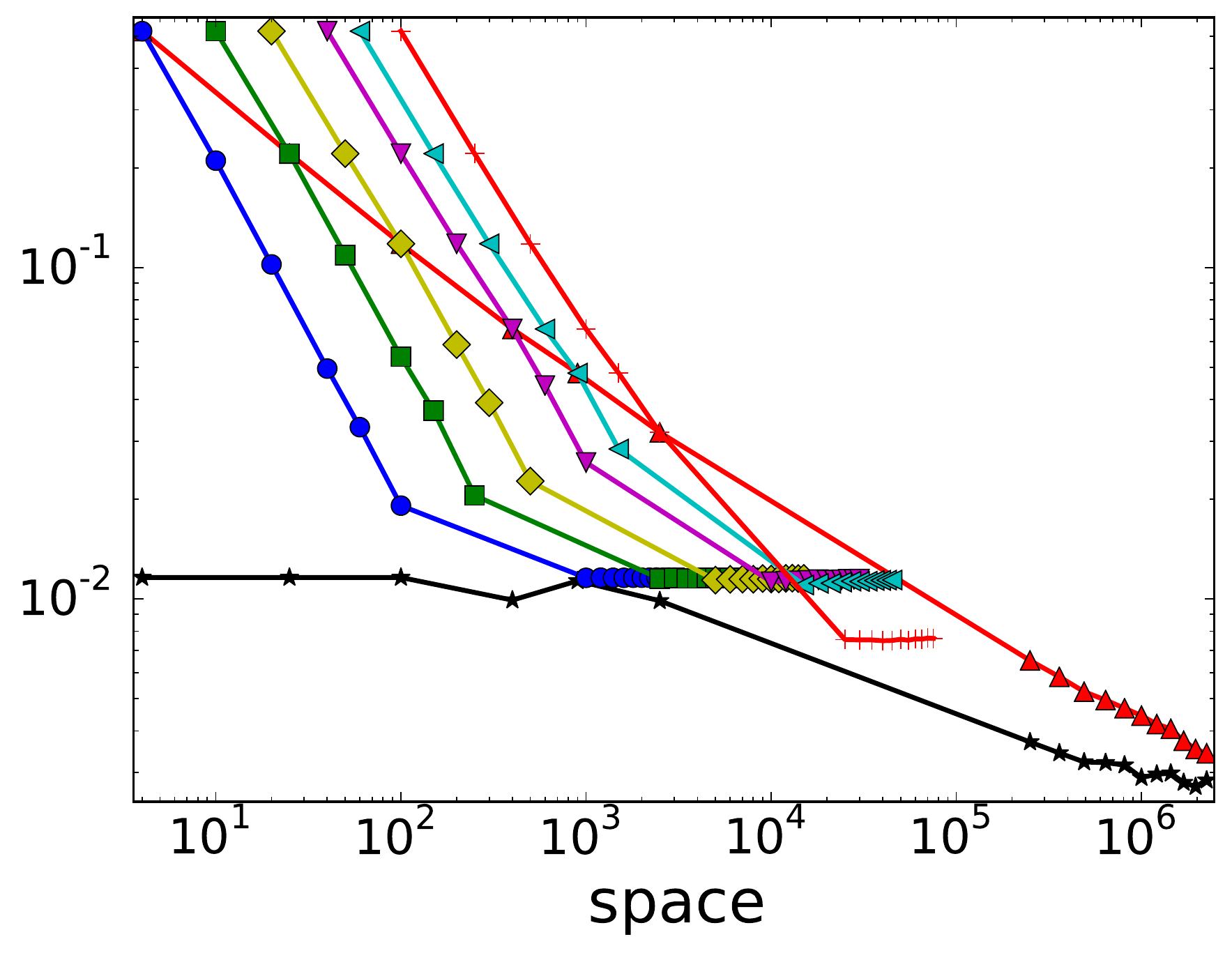}
%
\rotatebox{90}{\tiny \hspace{14mm}\textsf{Test time} (sec)}
\includegraphics[width=\figsize]{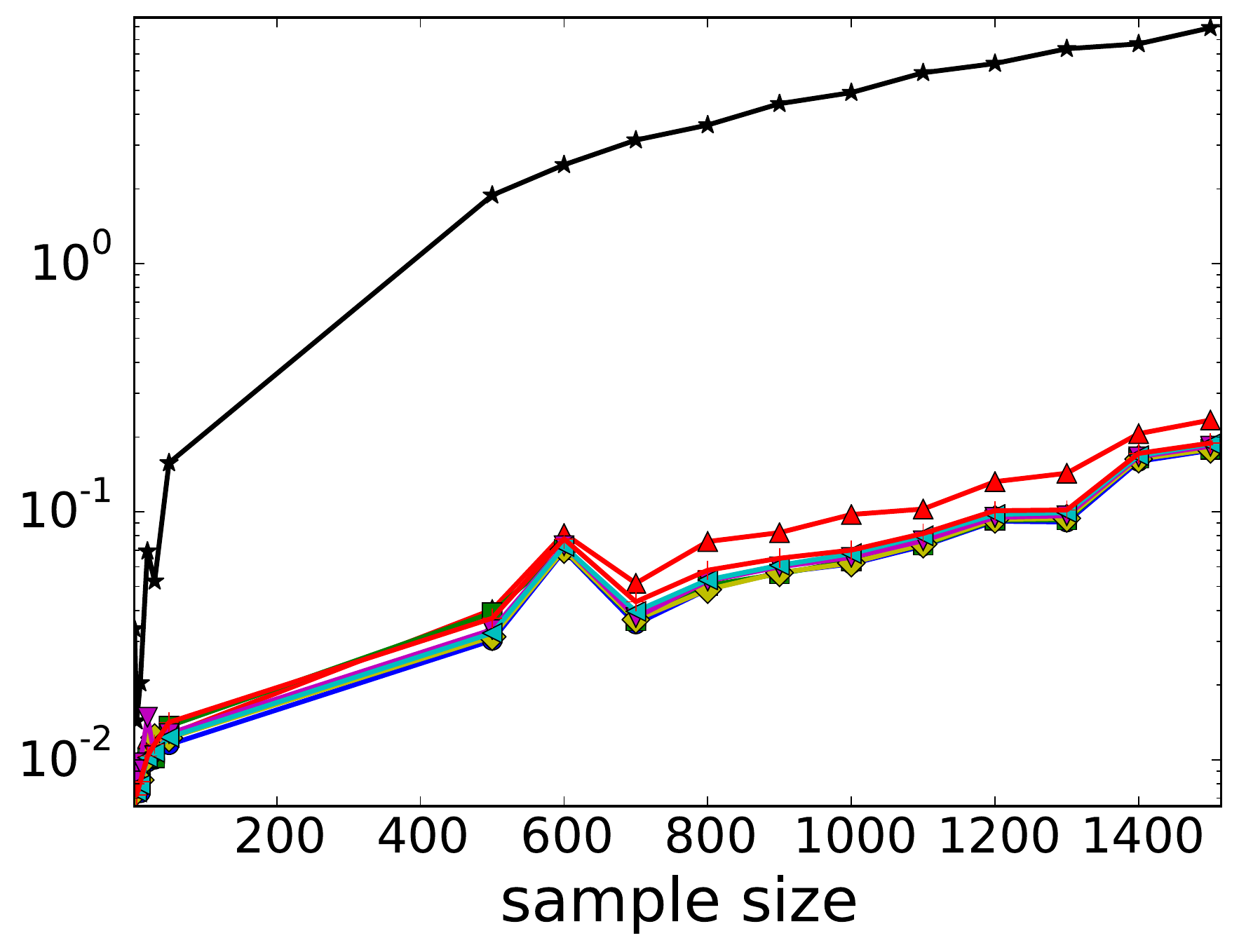}
%

\includegraphics[width=3.3\figsize]{hlegend.pdf}


\caption{{\small \textsc{CPU} dataset showing Error and Test and Train time versus Sample Size ($m$ or $c$) and Space.}} 


\label{fig:cpu}
\end{figure*}

\begin{figure*}[t!]
\rotatebox{90}{\tiny \hspace{10mm}\textsf{Kernel Frobenius Error}} 
\includegraphics[width=\figsize]{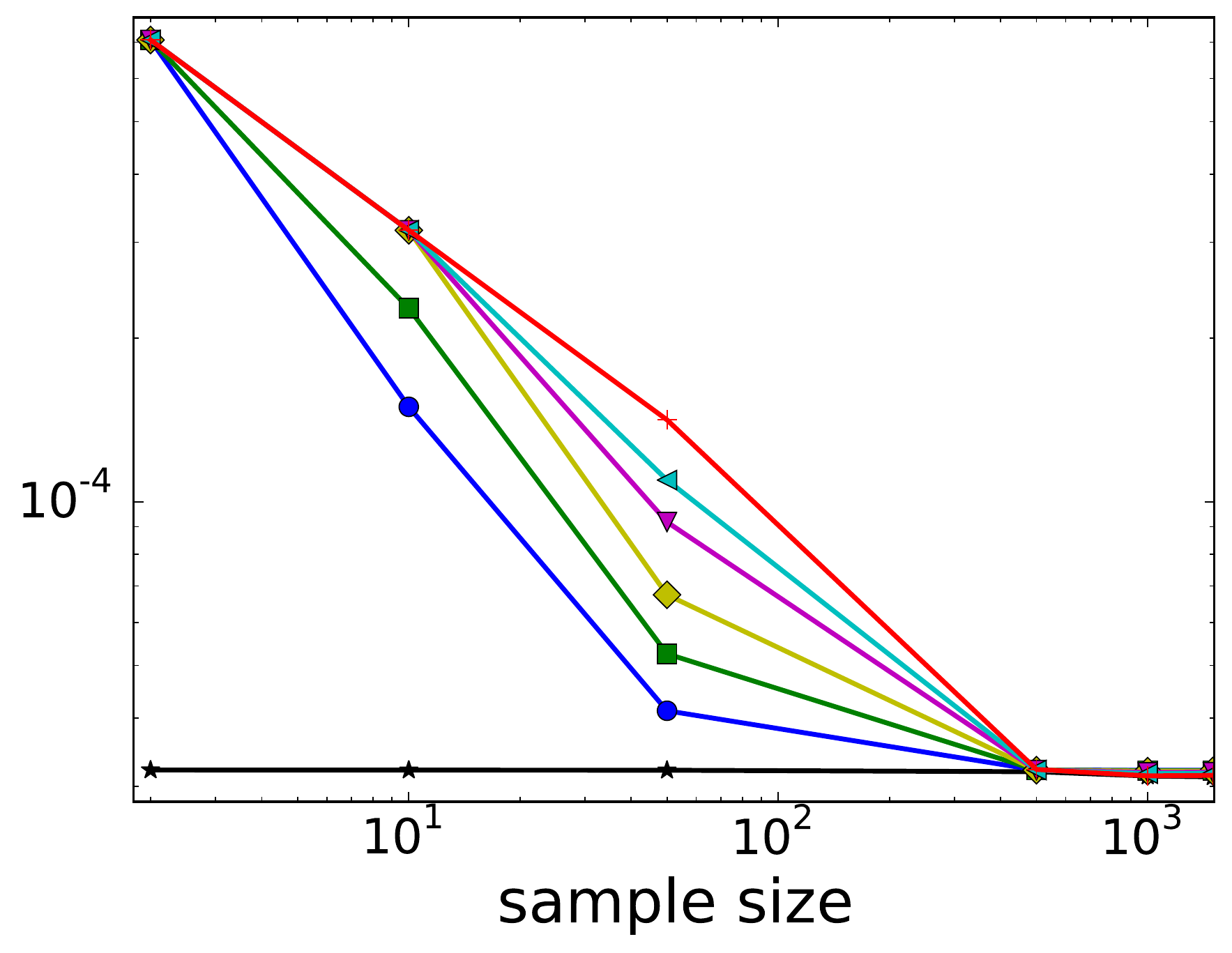}
%
\rotatebox{90}{\tiny \hspace{10mm}\textsf{Kernel Spectral Error}}
\includegraphics[width=\figsize]{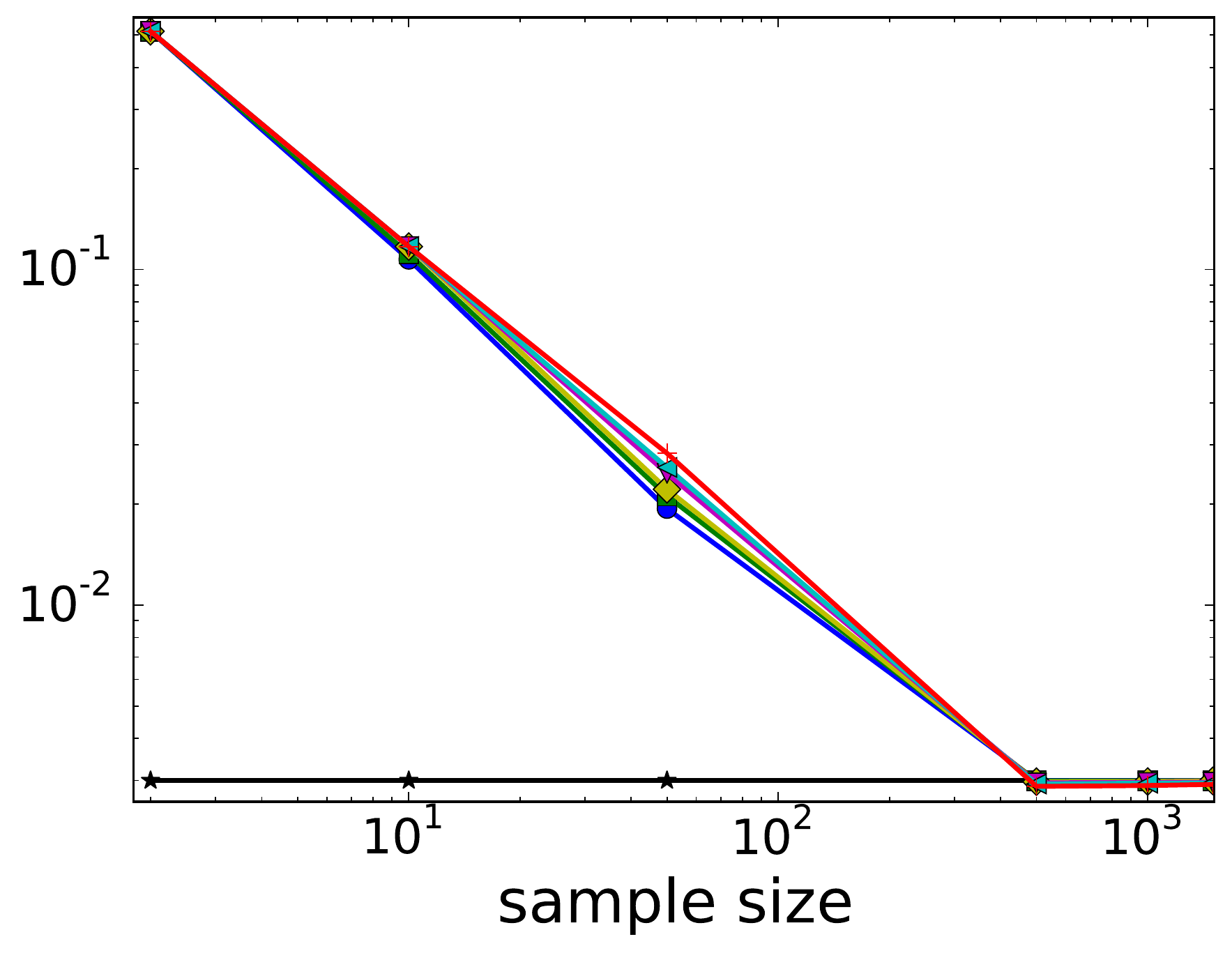}
%
\rotatebox{90}{\tiny \hspace{14mm}\textsf{Train time} (sec)}
\includegraphics[width=\figsize]{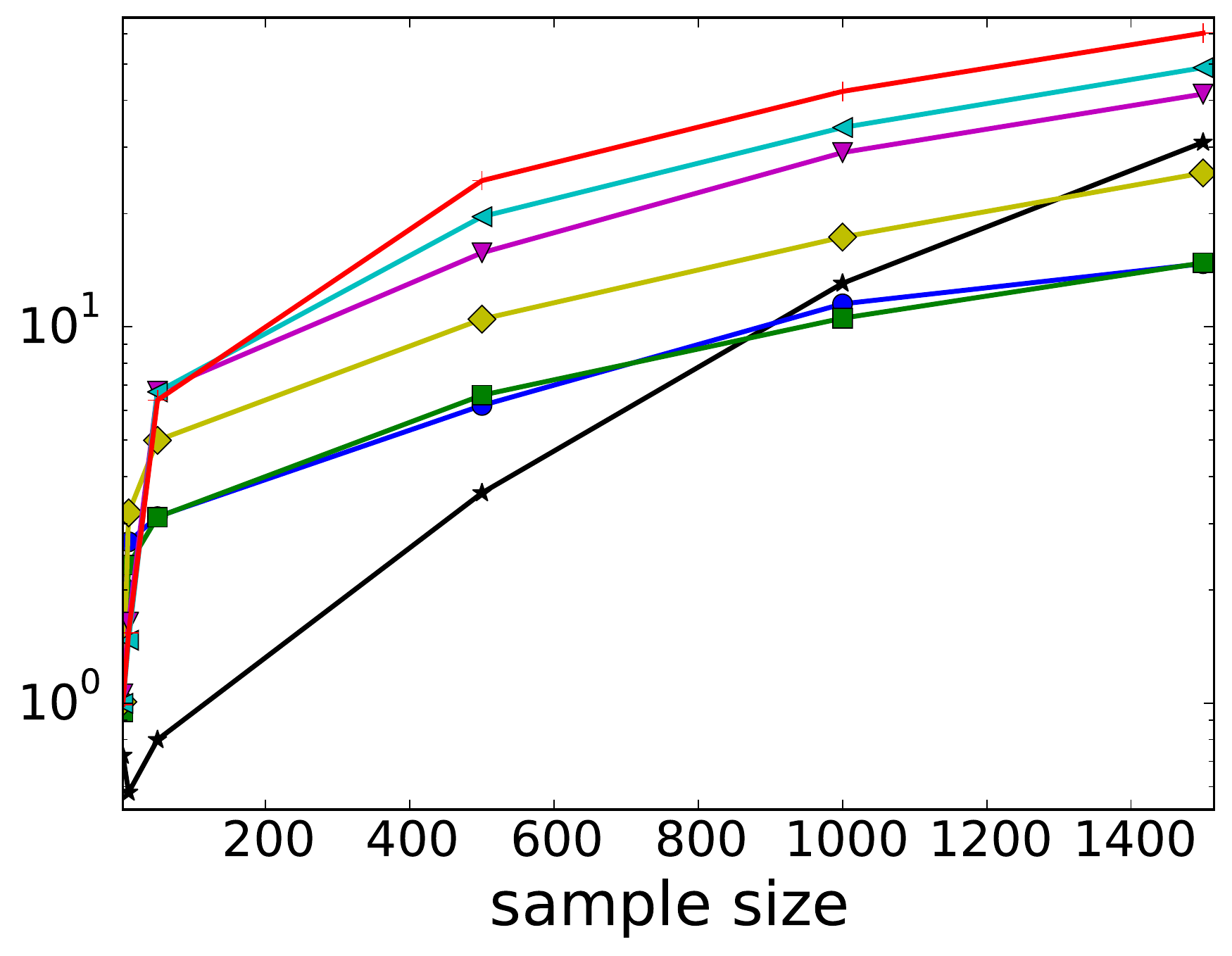}

\rotatebox{90}{\tiny \hspace{10mm}\textsf{Kernel Frobenius Error}}
\includegraphics[width=\figsize]{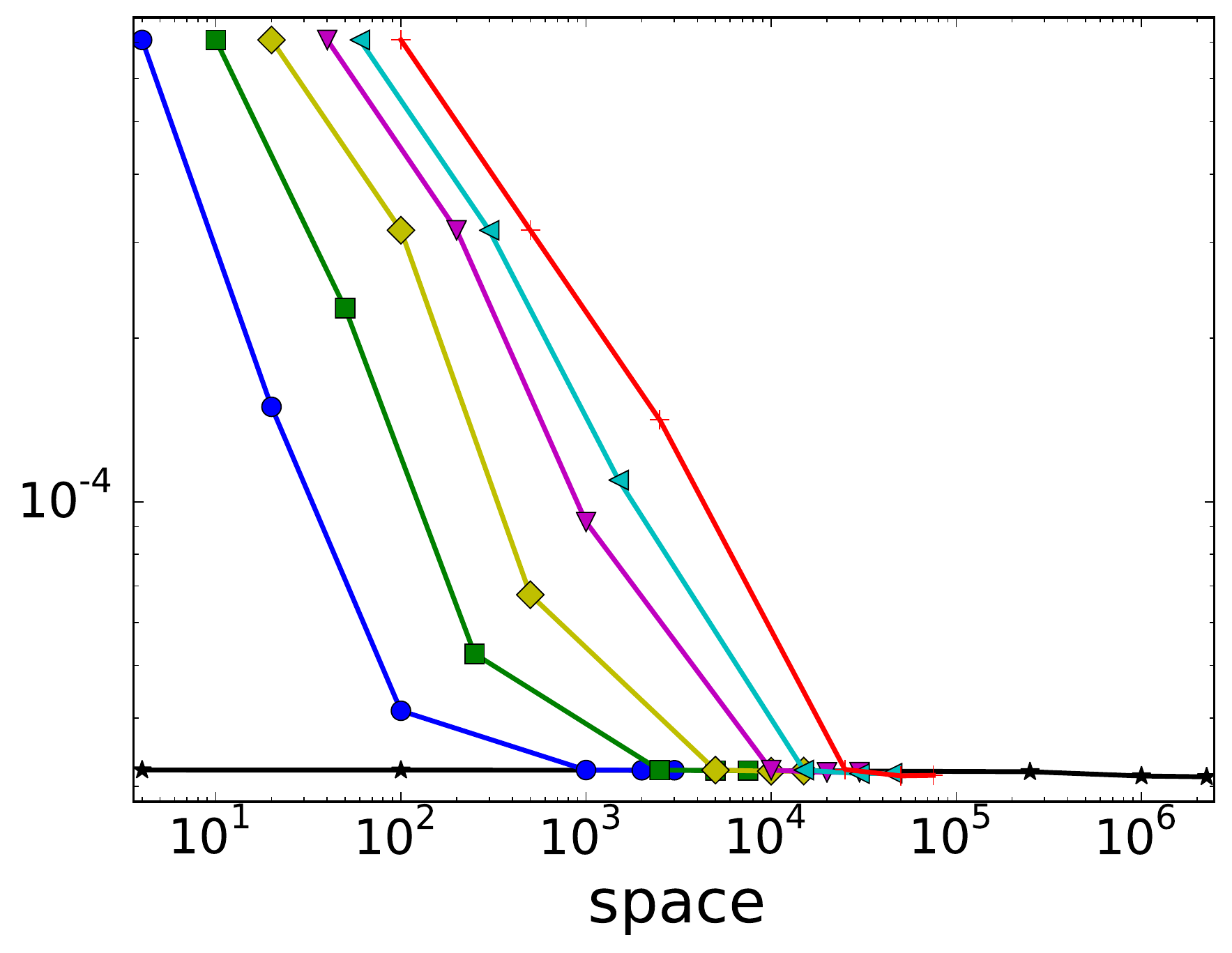}
%
\rotatebox{90}{\tiny \hspace{10mm}\textsf{Kernel Spectral Error}}
\includegraphics[width=\figsize]{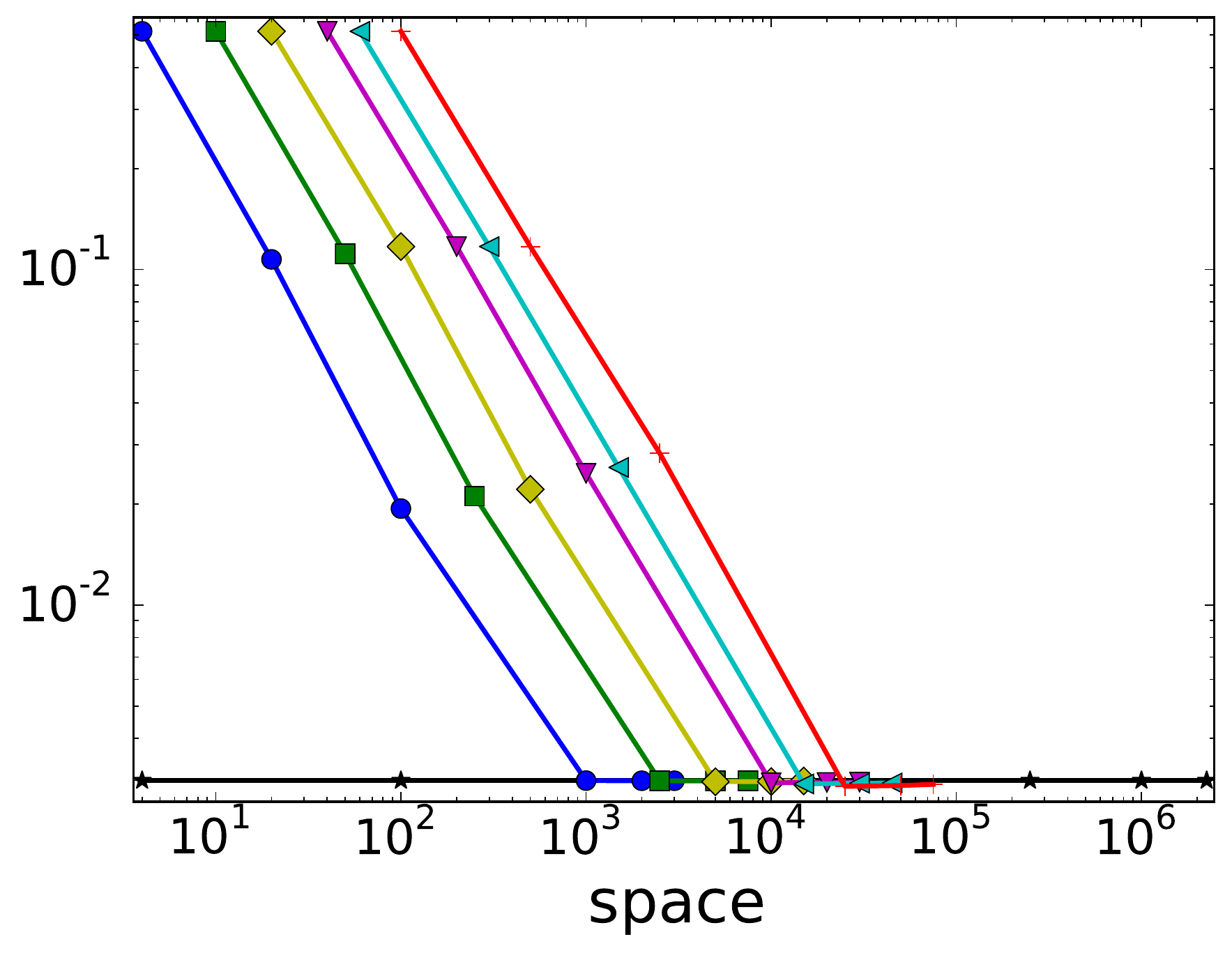}
%
\rotatebox{90}{\tiny \hspace{14mm}\textsf{Test time} (sec)}
\includegraphics[width=\figsize]{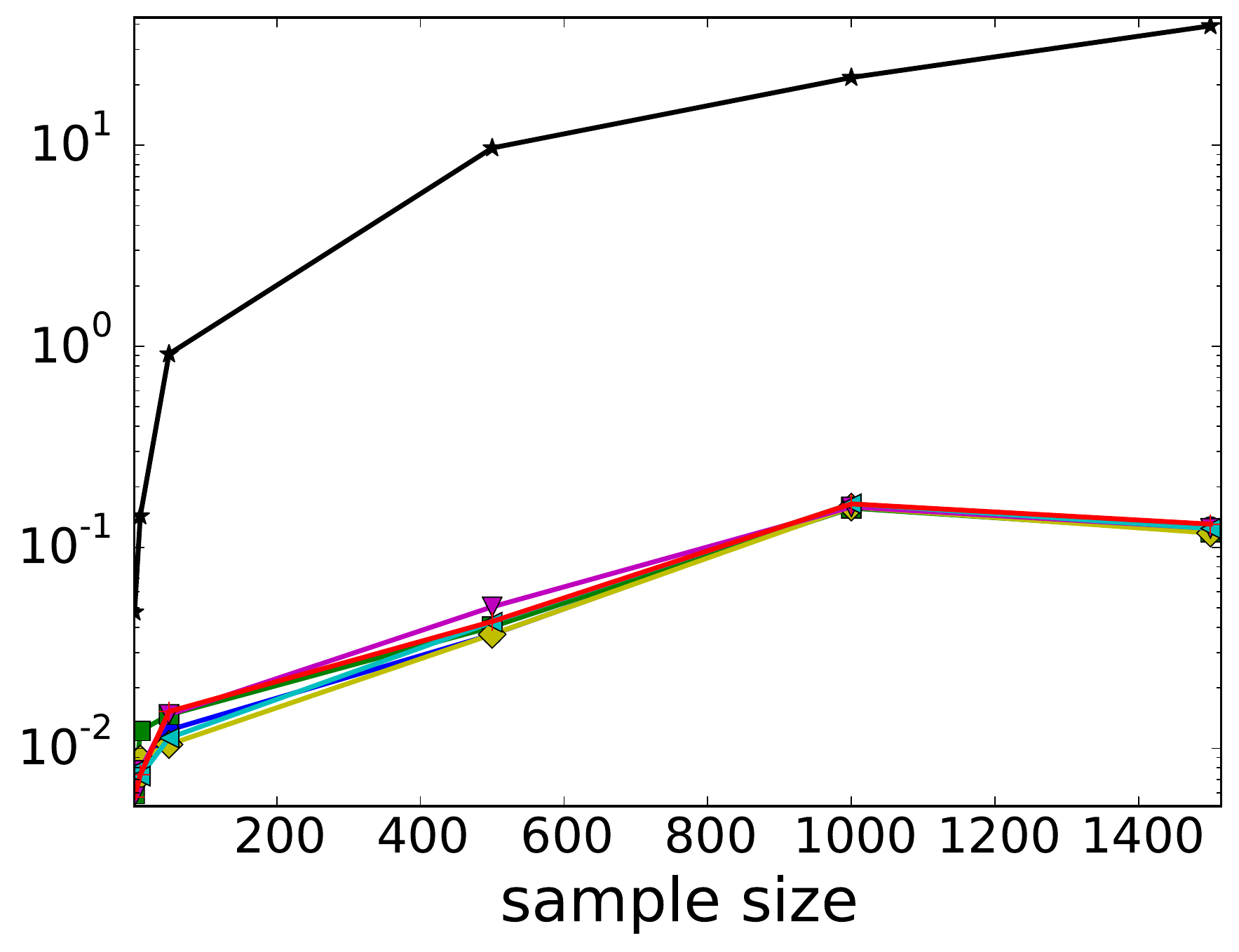}
%

\includegraphics[width=3.3\figsize]{hlegend.pdf}


\caption{{\small \textsc{Adult} dataset showing Error and Test and Train time versus Sample Size ($m$ or $c$) and Space.}} 



\label{fig:adult}
\end{figure*}

\begin{figure*}[t!]
\rotatebox{90}{\tiny \hspace{10mm}\textsf{Kernel Frobenius Error}} 
\includegraphics[width=\figsize]{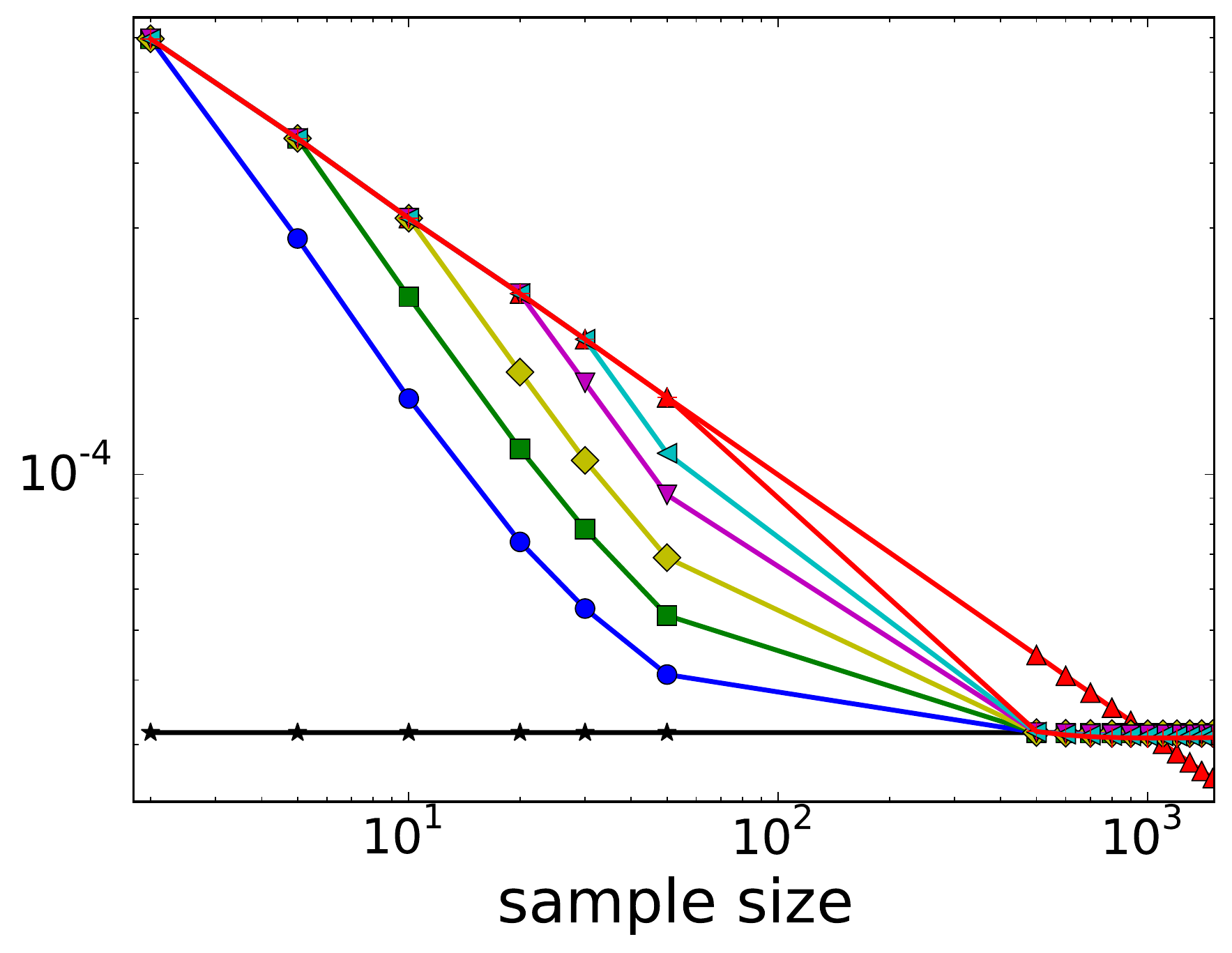}
%
\rotatebox{90}{\tiny \hspace{10mm}\textsf{Kernel Spectral Error}}
\includegraphics[width=\figsize]{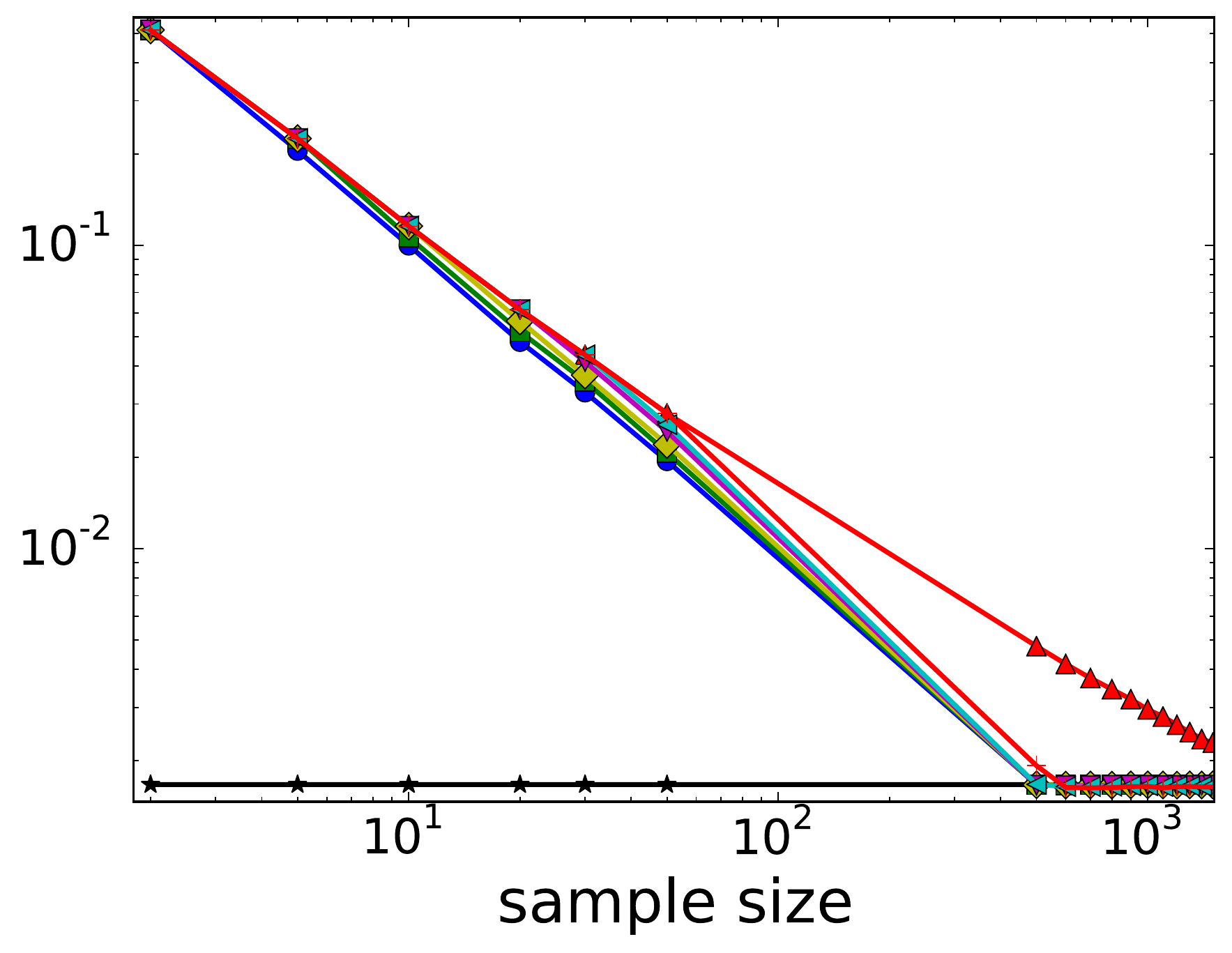}
%
\rotatebox{90}{\tiny \hspace{14mm}\textsf{Train time} (sec)}
\includegraphics[width=\figsize]{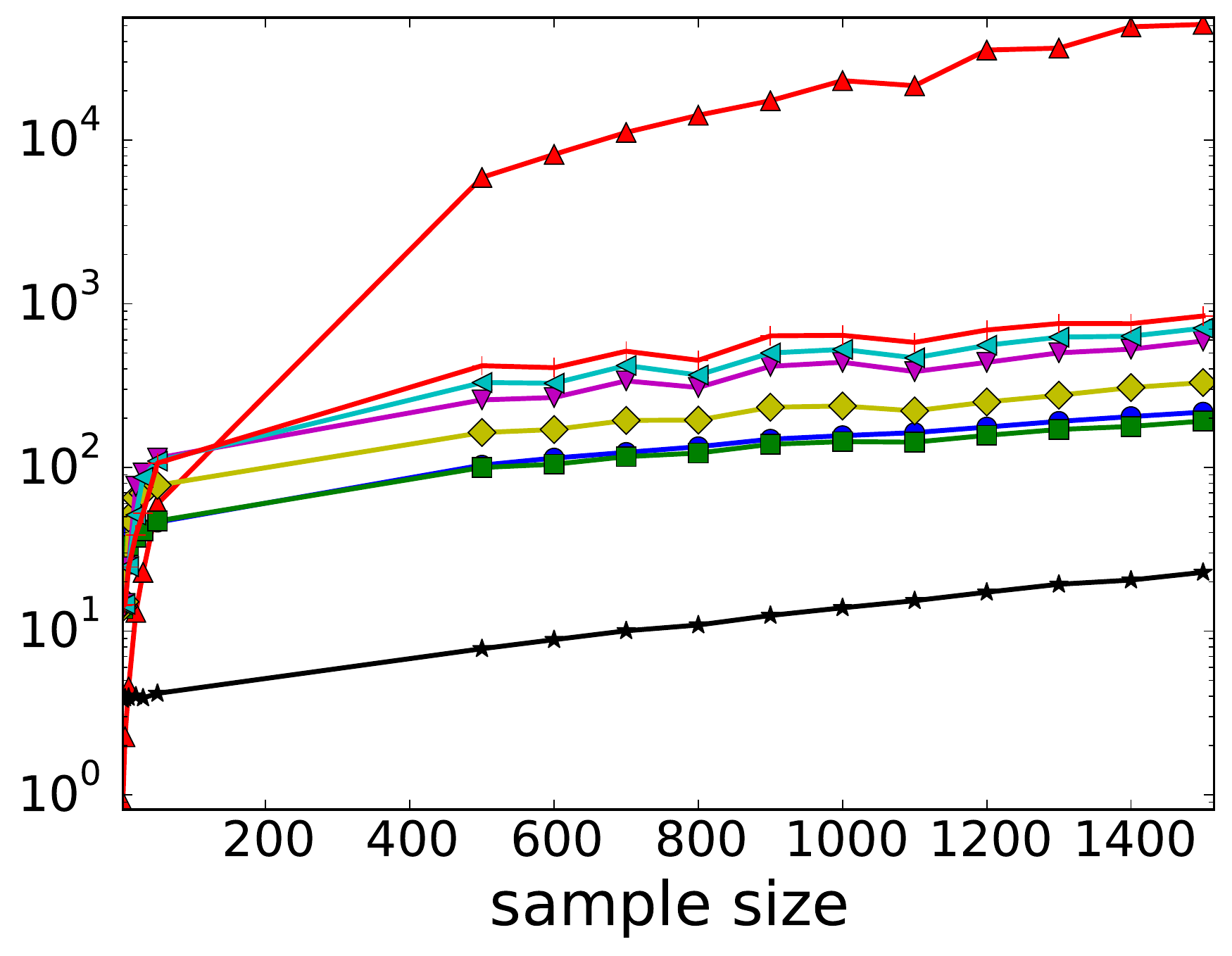}

\rotatebox{90}{\tiny \hspace{10mm}\textsf{Kernel Frobenius Error}}
\includegraphics[width=\figsize]{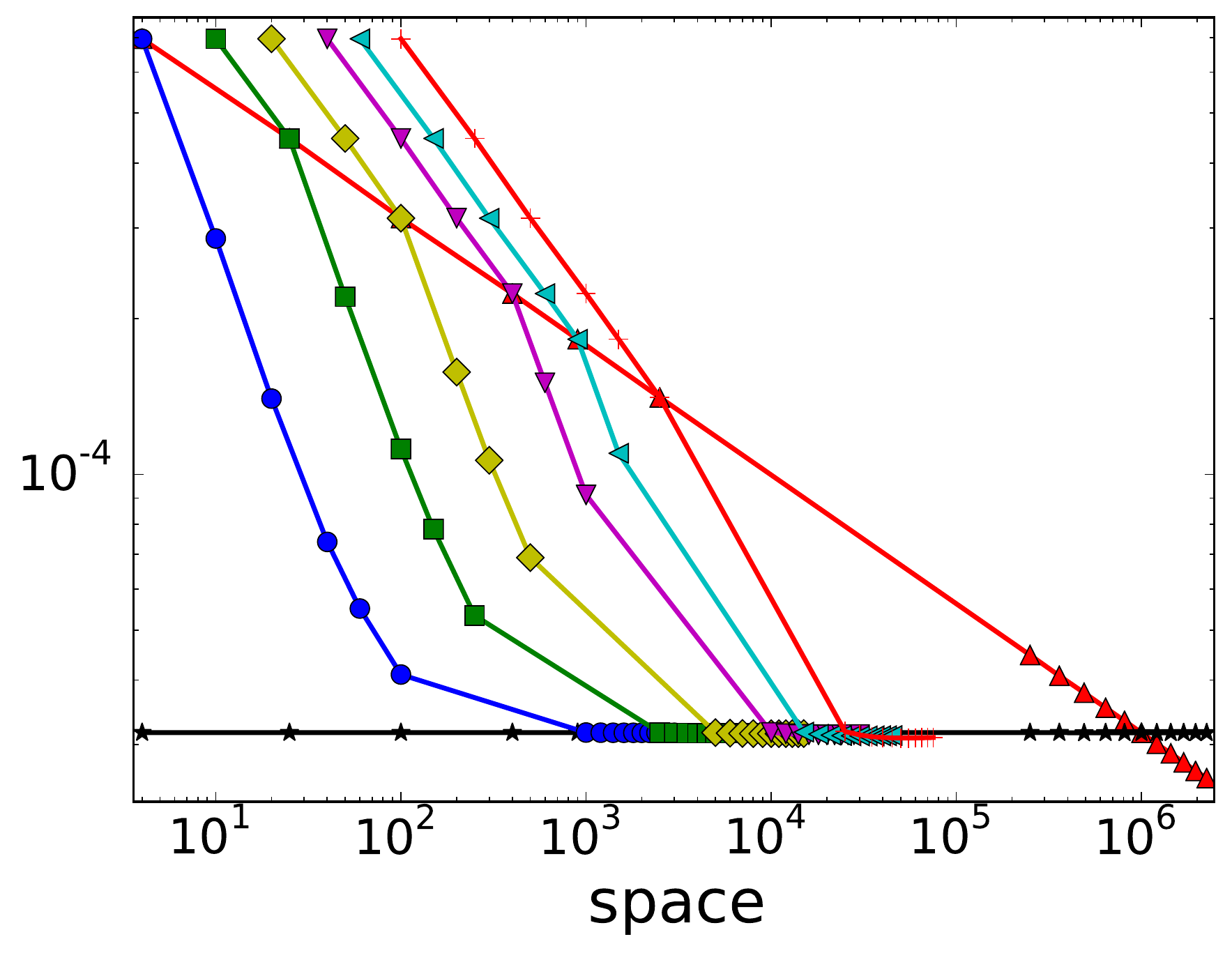}
%
\rotatebox{90}{\tiny \hspace{10mm}\textsf{Kernel Spectral Error}}
\includegraphics[width=\figsize]{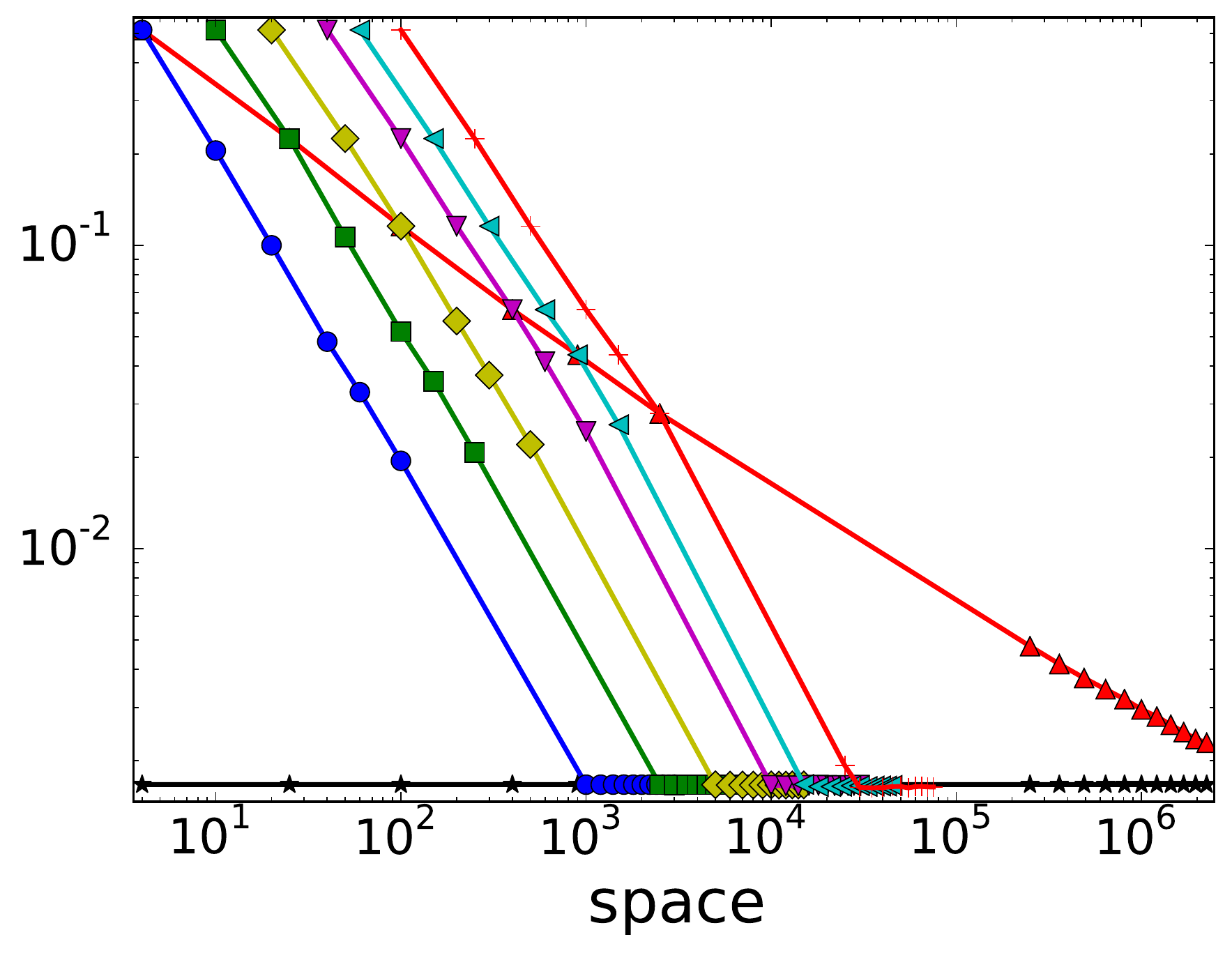}
%
\rotatebox{90}{\tiny \hspace{14mm}\textsf{Test time} (sec)}
\includegraphics[width=\figsize]{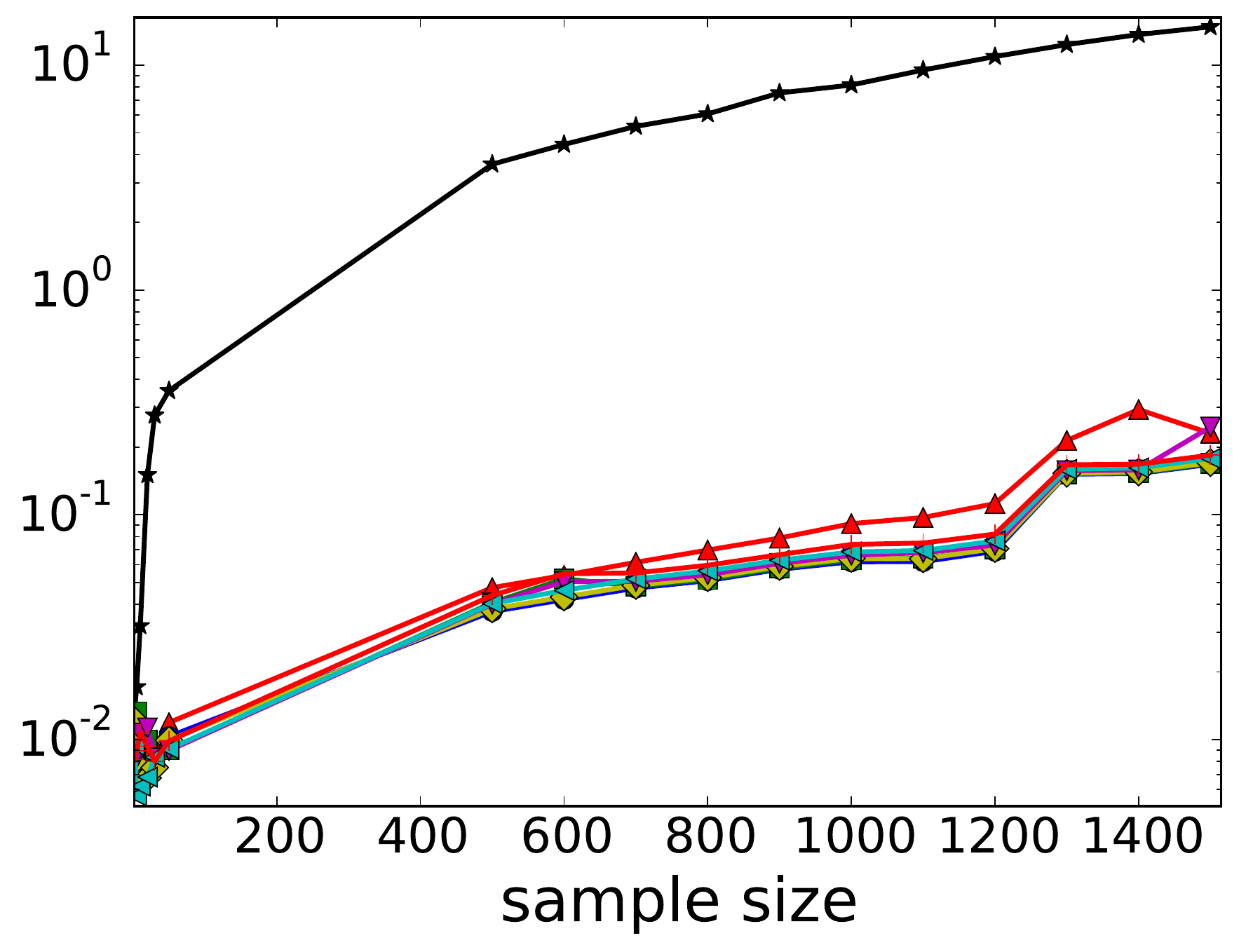}
%

\includegraphics[width=3.3\figsize]{hlegend.pdf}

\caption{{\small Results for \textsc{Forest} dataset.  
Row 1: \s{Kernel Frobenius Error} (left), \s{Kernel Spectral Error} (middle) and \textsc{Train Time} (right) vs{.} \textsc{Sample size}. 
Row 2: \s{Kernel Frobenius Error} (left), \s{Kernel Spectral Error} (middle) vs{.} \textsc{Space}, and \textsc{Test Time} vs{.} \textsc{Sample size} (right) }}

\label{fig:forest}
\end{figure*}

\paragraph{Error measures.}
We consider two error measures comparing the true gram matrix $G$ and an approximated gram matrix (constructed in various ways).  
The first error measure is \s{Kernel Spectral Error} $= \|G - G'\|_2/n$ which represents the worst case error.  
The second is \s{Kernel Frobenius Error} $= \|G - G'\|_F/n^2$ that represents the global error.  
We normalized the error measures by $1/n$ and $1/n^2$, respectively, so they are comparable across data sets. 
These measures require another pass on the data to compute, but give a more holistic view of how accurate our approaches are.

We measure the \textsc{Space} requirements of each algorithm as following: 
\SKPCA sketch has space $md+m\ell$, \Nyst is $c^2 + cd$, and RNCA is $m^2 + md$, where $m$ is the number of \RFM, and $c$ is the number of samples in \Nyst. 
In our experiments, we set $m$ and $c$ equally, calling these parameters \textsc{Sample Size}.  
Note that \textsc{Sample Size} and \textsc{Space} usage are different: both RNCA and Nystr\"om have \textsc{Space} quadratic in \textsc{Sample Size}, while for SKPCA it is linear.

\paragraph{Results.}
Figures \ref{fig:random_noisy}, \ref{fig:cpu}, \ref{fig:adult} and \ref{fig:forest} show log-log plots of results for \textsc{Random Noisy}, \textsc{CPU},  \textsc{Adult} and \textsc{Forest} datasets, respectively.    


For small \textsc{Sample Size} we observe that \Nyst performs quite well under all error measures, corroborating results reported by Lopez \etal~\cite{lopez2014randomized}.  However, all methods have a \emph{very small error range}, typically less than $0.01$.  
For \s{Kernel Frobenius Error} we typically observe a cross-over point where RNCA and often most versions of SKPCA have better error for that size.  
Under \s{Kernel Spectral Error} we often see a cross-over point for SKPCA, but not for RNCA.  
We suspect that this is related to how FD only maintains the most dominate directions while ignoring other (potentially spurious) directions introduced by the \RFM. 
In general, SKPCA has as good or better error than RNCA for the same size, with smaller size being required with smaller $\ell$ values.  This difference is more pronounced in \textsc{Space} than \textsc{Sample Size}, where our theoretical results expect a polynomial advantage.

In timing experiments, especially \textsc{Train time} we see \SKPCA has a very large advantage.  
As a function of \textsc{Sample Size} RNCA is the slowest for \textsc{Train time}, and \Nyst is the slowest for \textsc{Test time} by several orders of magnitude.  In both cases all versions of SKPCA are among the fastest algorithms.  
For the \textsc{Train time} results, RNCA's slow time is dominated by summing $n$ outer products, of dimensions $m \times m$.  This is avoided in SKPCA by only keeping the top $\ell$ dimensions, and only requiring similar computation on the order of $\ell \times m$, where typically $\ell \ll m$.  
\Nyst only updates the $c \times c$ gram matrix when a new point replaces an old one,  expected $c \log n$ times.  

\Nyst is comparatively very slow in \textsc{Test time}.   It
computes a new row and column of the gram matrix, and projects onto this space, taking $O(cd + c^2)$ time.  
Both RNCA and SKPCA avoid this by directly computing an $m$ dimensional representation of a test data point in $O(dm)$ time.  
Recall we precompute the eigen-structure for RNCA and \Nyst, whereas SKPCA maintains it at all times, so if this step were counted, SKPCA's advantage here would be even larger.


\paragraph{Summary.}
Our proposed method SKPCA has superior timing and error results to RNCA, by sketching in the kernel feature space.  Its error is typically a bit worse than a \Nyst approach, but the difference is quite small, and SKPCA is far superior to \Nyst in \textsc{Test time}, needed for any data analysis.


\newpage

\subsubsection*{References}

\begingroup

\renewcommand{\section}[2]{}%
\bibliographystyle{plain}
\bibliography{kpca}

\endgroup


\end{document}